\documentclass[9pt]{article}
\usepackage{epsf,epsfig,verbatim,amssymb,amsmath,array,cite,amsthm,subfigure,multicol,multirow}

\usepackage{psfrag,bm,xspace}
\newtheorem{lemma}{Lemma}
\newtheorem{defi}{Definition}
\newtheorem{thm}{Theorem}

\newtheorem{cor}{Corollary}
\newtheorem{fact}{Fact}

\begin{document}
\bibliographystyle{ieeetr}
\newcommand{\ring}[1]{\mathbb{Z}_{#1}}
\newcommand{\ringr}[1]{\mathbb{Z}_{p^r}^{#1}}
\newcommand{\mc}[1]{\mathcal{#1}}
\newcommand{\typset}[2]{A_{\epsilon_{#1}}^{n}(#2)}
\newcommand{\gf}[2]{\mathbb{F}_{#1}^{#2}}
\newcommand{\ringrb}{\mathbb{Z}_{p^b}^{r_b}}

\title{\LARGE \bf
Distributed Source Coding using Abelian Group Codes}

\author{Dinesh Krithivasan and S. Sandeep Pradhan\thanks{This work was
    supported by NSF grant (CAREER) CCF-0448115. Part of this work was presented in ISIT 2008 held at Toronto, Canada.}, \\
Department of Electrical Engineering and Computer Science, \\
University of Michigan, Ann Arbor, MI 48109, USA \\
email: {\tt\small dineshk@umich.edu, pradhanv@eecs.umich.edu}}


\maketitle \thispagestyle{empty} \pagestyle{plain}

\begin{abstract}
In this work, we consider a distributed source coding problem with a joint distortion criterion depending on the sources and the reconstruction. This includes as a special case the problem of computing a function of the sources to within some distortion and also the classic Slepian-Wolf problem \cite{slepian-wolf}, Berger-Tung problem \cite{berger-tung}, Wyner-Ziv problem \cite{wyner-ziv}, Yeung-Berger problem \cite{yeung-berger} and the Ahlswede-Korner-Wyner problem \cite{ahlswede-korner,wyner75}. While the prevalent trend in information theory has been to prove achievability results using Shannon's random coding arguments, using structured random codes offer rate gains over unstructured random codes for many problems. Motivated by this, we present a new achievable rate-distortion region (an inner bound to the performance limit) for this problem for discrete memoryless sources based on ``good'' structured random nested codes built over abelian groups. We demonstrate rate gains for this problem over traditional coding schemes using random unstructured codes. For certain sources and distortion functions, the new rate region is strictly bigger than the Berger-Tung rate region, which has been the best known achievable rate region for this problem till  now.  Further, there is no known unstructured random coding scheme that achieves these rate gains. Achievable performance limits for single-user source coding using abelian group codes are also obtained as parts of the proof of the main coding theorem. As a corollary, we also prove that nested linear codes achieve the Shannon rate-distortion bound in the single-user setting. Note that while group codes retain some structure, they are more general than linear codes which can only be built over finite fields which are known to exist only for certain sizes.

\end{abstract}

\section{Introduction} \label{sec:intro}

A large number of problems in multi-user information theory fall under the general setup of distributed source coding. The most general framework for a distributed source coding problem consists of a set of encoders which observe different correlated components of a vector source and communicate their quantized observations to a central decoder through a rate-constrained noiseless communication link. The decoder is interested in reconstructing these observations or some function of them to within some distortion as measured by a fidelity criterion. The goal is to obtain a computable single-letter characterization of the performance limits measured by the rates of transmission and the distortions achieved. Such a formulation finds wide applications in many areas of communications such as sensor networks and distributed computing.

There is a vast body of work that addresses this problem, and solutions have been obtained in a variety of special cases \cite{wyner75,ahlswede-korner,wyner-ziv,berger77,korner-marton,gelfand-pinsker,yamamoto,yamamoto-itoh,flynn-gray,yeung-berger,orlitsky-roche}\footnote{We have restricted our attention to discrete memoryless sources. There has been a lot of research activity in the literature in the case of continuous-alphabet sources. Those works are not included in the reference list for conciseness. Please see the references in our earlier work on Gaussian sources \cite{dinesh-pradhan} for a more complete list.}. All of the existing works use the following canonical encoding strategy. Each encoder has two operations, implemented sequentially,  each of which is a many-to-one mapping: (a) quantization and (b) binning. In quantization, typically, neighboring source sequences are assigned a codeword, whereas in binning, a widely separated set of codewords is assigned a single index which is transmitted over a noise-free channel to the decoder. The decoder looks for the most likely tuple of codewords, one from each source, and then obtains a reconstruction as a function of this tuple of codewords.

In most of these works, existence of good encoders and decoder is shown by using random vector quantization followed by random independent binning of the quantizer codebooks. The best known inner bound to the performance limit that uses this  approach is the Berger-Tung inner bound. It has been shown in the literature that this is optimal in several  cases. The work of Korner and Marton \cite{korner-marton}, however, is an exception and looks at a special case of the problem involving a pair of doubly symmetric binary sources and near lossless reconstruction of the sample-wise logical XOR function of the source sequences. They considered an encoding strategy where the first operation is an identity transformation. For the second operation, they consider random structured binning of the spaces of source sequences and show optimality.  Further, the binning of two spaces is done in a ``correlated'' fashion using a binary linear code.  In the present paper, we build on this work, and present a new achievable rate region for the general distributed source coding  problem and demonstrate an encoding scheme that achieves this rate region by using random coding on structured code ensembles. In this approach, we consider the case where the sources are stationary discrete memoryless and the reconstruction is with respect to a general single-letter fidelity criterion.  The novelty of our approach lies in an unified treatment of the problem that works for any arbitrary function that the decoder is interested in reconstructing. Further, our approach relies on the use of nested group codes for encoding. The binning operation of the encoders are done in a ``correlated'' manner as dictated by these structured codes. This use of ``structured quantization followed by correlated binning'' is in contrast to the more prevalent ``quantization using random codes followed by independent binning'' in distributed source coding. This approach unifies all the known results in distributed source coding such as the Slepian-Wolf problem \cite{slepian-wolf}, Korner-Marton problem \cite{korner-marton}, Wyner-Ahlswede-Korner problem \cite{wyner75,ahlswede-korner}, Wyner-Ziv problem \cite{wyner-ziv}, Yeung-Berger problem \cite{yeung-berger} and Berger-Tung problem \cite{berger-tung}, under a single framework while recovering their respective rate regions. Moreover, this approach performs strictly better than the standard Berger-Tung based approach for certain source distributions. As a corollary, we show that nested linear codes can achieve the Shannon rate-distortion function in the single source point-to-point setting.  A similar correlated binning strategy for reconstructing linear functions of jointly Gaussian sources with mean squared error criterion was presented in \cite{dinesh-pradhan}. The present work develops a similar framework based on group codes.

This rate region is developed using the following two new ideas. First, we use the fact that any abelian group is isomorphic to the direct sum of primary cyclic groups to enable the decomposition of the source into its constituent ``digits'' which are then encoded sequentially.  Second, we show that, although group codes may not approach the Shannon rate-distortion function in a single source point-to-point setting, it is possible to construct non-trivial group codes which contain a code that approaches it. Using these two ideas, we provide an all-group-code solution to the problem and characterize an inner bound to the performance limit using single-letter information quantities. We also demonstrate the superiority of this approach over the conventional coding approach based on unstructured random codes for the case of reconstructing the modulo-$2$ sum of correlated binary sources with Hamming distortion.

Special cases of the general problem of distributed source coding with joint distortion criterion have been studied before. The minimum rate at which a source $X$ must be transmitted for the decoder to enable lossless reconstruction of a bivariate function $F(X,Y)$ with perfect side information $Y$ was determined in \cite{orlitsky-roche}. The case when the two communicators are allowed to exchange two messages was also considered. A two terminal interactive distributed source coding problem where the terminals exchange potentially infinite number of messages with the goal of reconstructing a function losslessly was studied in \cite{ma-ishwar}. A similar problem of function computation from various sensor measurements in a wireless network was studied in \cite{giridhar-kumar} in the context of a packet collision model.

Prior Work on Group Codes: Good codes over groups have been studied extensively in the literature when the order (size) of the group is a prime which enables the group to have a field structure. Such codes over Galois fields have been studied for the purpose of packing and covering (see \cite{coveringcodes,conway-sloane} and the references therein). Two kinds of packing problems have received attention in the literature: a) combinatorial rigid packing  and b) probabilistic soft packing, i.e., achieving the capacity of symmetric channels. Similarly, covering problems have been studied in two ways: a) combinatorial complete covering and (b) probabilistic almost covering, i.e., achieving the rate-distortion function of symmetric sources with Hamming distortion. Some of the salient features of these two approaches have been studied in \cite{loeliger-ensemble}. In the following we give a sample of works in the direction of probabilistic packing and covering. Elias \cite{elias} showed that linear code achieve the capacity of binary symmetric channels. A reformulation of this result can be used to show \cite{korner-marton} that linear codes can be used to losslessly compress any discrete source down to its entropy. Dobrushin \cite{dobrushin63} showed that linear codes achieve the random coding error exponent while Forney and Barg \cite{forney-barg} showed that linear codes also achieve the expurgated error exponent. Further, these results have been shown to be true for almost all linear codes. Gallager \cite{gallager-book} shows that binary linear codes succeeded by a nonlinear mapping can approach the capacity of any discrete memoryless channel. It follows from Goblick's work \cite{goblick,delsarte,cohen} on the covering radius of linear codes that linear codes can be used to achieve the rate distortion bound for binary sources with Hamming distortion. Blinovskii \cite{blinovskii} derived upper and lower bounds on the covering radius of linear codes and also showed that almost all linear codes (satisfying rate constraints) are good source codes for binary sources with Hamming distortion. If the size of the finite field is sufficiently large, it was shown that in \cite{junchen} that linear codes followed by a nonlinear mapping can achieve the rate distortion bound of a discrete memoryless source with arbitrary distortion measure. Wyner \cite{wyner74} derived an algebraic binning approach to provide a simple derivation of the Slepian-Wolf \cite{slepian-wolf} rate region for the case of correlated binary sources. Csiszar \cite{csiszar82} showed the existence of universal linear encoders which attain the best known error exponents for the Slepian-Wolf problem derived earlier using nonlinear codes. In \cite{shamai-verdu-zamir,zamirmulti}, nested linear codes were used for approaching the Wyner-Ziv rate-distortion function for the case of doubly symmetric binary source and side information with Hamming distortion.  Random structured codes have been used in other related multiterminal communication problems \cite{nazer,zamir07,philosof-zamir} to get performance that is superior to that obtained by random unstructured codes. In \cite{muramatsu-miyake}, a coding scheme based on sparse matrices and ML decoding was presented that achieves the known rate regions for the Slepian-Wolf problem, Wyner-Ziv problem and the problem of lossless source coding with partial side information.

Codes over general cyclic groups were first studied by Slepian \cite{slepian-group} in the context of signal sets for the Gaussian channel. Forney \cite{forney91} formalized the concept of geometrically uniform codes and showed that many known classes of good signal space codes were geometrically uniform. Biglieri and Elia \cite{biglieri-elia1} addressed the problem of existence of group codes for the Gaussian channel as defined by Slepian. Forney and Loeliger \cite{forney93,loeliger96} studied the state space representation of group codes and derived trellis representations which were used to build convolutional codes over abelian groups. An efficient algorithm for building such minimal trellises was presented in \cite{vazirani-rajan}. Loeliger \cite{loeliger-signal} extended the concept of the $M$-PSK signal set matched to the $M$-ary cyclic group to the case of matching general signal sets with arbitrary groups. Building codes over abelian groups with good error correcting properties was studied in \cite{berman}. The distance properties of group codes have also been extensively studied. In \cite{forney92,biglieri-elia2,interlando-elia}, bounds were derived on the minimum distance of group codes and it was also shown that codes built over nonabelian groups have asymptotically bad minimum distance behavior. Group codes have also been used to build LDPC codes with good distance properties \cite{tanner}. The information theoretic performance limits of group codes when used as channel codes over symmetric channels was studied in \cite{como-fagnani}. Similar analysis for the case of turbo codes and geometrically uniform constellations was carried out in \cite{garin-fagnani}. In \cite{ahlswede71},  Ahlswede established the achievable capacity using group codes for several classes of channels and showed that in general, group codes do not achieve the capacity of a general discrete memoryless channel. Sharper results were obtained for the group codes capacity and their upper bounds in \cite{ahlswede-gemma1, ahlswede-gemma2}.

The paper is organized as follows. In Section \ref{sec:probdef}, we define the problem formally and present known results for the problem. In Section \ref{sec:groupintro}, we present an overview of the properties of groups in general and cyclic groups in particular that shall be used later on. We motivate our coding scheme in Section \ref{sec:motivation}. In Section \ref{sec:definitions}, we define the various concepts used in the rest of the paper. In Section \ref{sec:codingthm}, we present our coding scheme and present an achievable rate region for the problem defined in Section \ref{sec:probdef}. Section \ref{sec:specialcases} contains the various corollaries of the theorem presented in Section \ref{sec:codingthm}. These include achievable rates for lossless and lossy source coding while using group codes. We also present achievable rates using group codes for the problem of function reconstruction. Most of the proofs are presented in the appendix. In Section \ref{sec:examples}, we demonstrate the application of  our coding theorem to various problems. We conclude the paper with some comments in Section \ref{sec:conclusions}.

A brief overview of the notation used in the paper is given below. Random variables are denoted by capital letters such as $X,Y$ etc. The alphabet over which a discrete random variable $X$ takes values will be indicated by $\mc{X}$. The cardinality of a discrete set $\mc{X}$ is denoted by $|\mc{X}|$. For a random variable $X$ with distribution $p_X(\cdot)$, the set of all $n$-length strongly $\epsilon$-typical sequences are denoted by $A_{\epsilon}^n(X)$ \cite{csiszarbook}. On most occasions, the subscript and superscript are omitted and their values should be clear from the context. For a pair of jointly distributed random variables $X,Y$ with distribution $p_{X,Y}(\cdot,\cdot)$, the set of all $n$-length $y^n$-sequences jointly $\epsilon$-typical with a given $x^n$ sequence is denoted by the set $A_{\epsilon}^{n}(x^n)$.

\section{Problem Definition and Known Results} \label{sec:probdef}

Consider a pair of discrete random variables $(X,Y)$ with joint distribution $p_{XY}(\cdot,\cdot)$. Let the alphabets of the random variables  $X$ and $Y$ be $\mc{X}$ and $\mc{Y}$ respectively. The source  sequence $(X^n,Y^n)$ is independent over time and has the product distribution  $Pr((X^n,Y^n)=(x^n,y^n))=\prod_{i=1}^{n} p_{XY}(x_i,y_i)$. We consider the following distributed source coding problem. The two components of the source are observed by two encoders which do not communicate with each other. Each encoder communicates a compressed version of its input through a noiseless channel to a joint decoder. The decoder is interested in reconstructing the sources with respect to a general fidelity criterion. Let $\hat{\mc{Z}}$ denote the reconstruction alphabet, and the fidelity criterion is characterized by a mapping: $d: \mc{X} \times \mc{Y} \times \hat{\mc{Z}} \rightarrow \Bbb{R}^+$. We restrict our attention to additive distortion measures, i.e., the distortion among three $n$-length sequences $x^n$, $y^n$ and $\hat{z}^n$ is given by
\begin{align}
\hat{d}(x^n,y^n,\hat{z}^n) &\triangleq \frac{1}{n} \sum_{i=1}^{n} d(x_i,y_i,\hat{z}_i).
\end{align}

In this work, we will concentrate on the above distributed source coding problem (with one distortion constraint), and provide an information-theoretic inner bound to the optimal rate-distortion region.

\begin{defi} \label{defi:probdefi1}
Given a discrete source with joint distribution $p_{XY}(x,y)$ and a distortion function $d(\cdot,\cdot,\cdot)$, a transmission system with parameters $(n,\theta_1,\theta_2,\Delta)$ is defined by the set of mappings
\begin{equation}
f_1 \colon \mc{X}^n \rightarrow \{1,\dots,\theta_1\}, \qquad f_2 \colon \mc{Y}^n \rightarrow \{1,\dots,\theta_2\}
\end{equation}
\begin{equation}
g \colon \{1,\dots,\theta_1\} \times \{1,\dots,\theta_2\} \rightarrow \hat{\mc{Z}}^n
\end{equation}
such that the following constraint is satisfied.
\begin{align}
\Bbb{E} (\hat{d}(X^n,Y^n,g(f_1(X^n),f_2(Y^n)))) &\leq \Delta.
\end{align}
\end{defi}

\begin{defi} \label{defi:probdefi2}
We say that a tuple $(R_1,R_2,D)$ is achievable if $\forall \epsilon > 0$, $\exists$ for all sufficiently large $n$ a transmission system with parameters $(n,\theta_1,\theta_2,\Delta)$ such that
\begin{equation}
\frac{1}{n} \log \theta_i \leq R_i + \epsilon \quad \mbox{for }i = 1,2 \quad \Delta \leq D + \epsilon.
\end{equation}
The performance limit is given by the optimal rate-distortion region $\mathcal{RD}$ which is defined as the set of all achievable tuples $(R_1,R_2,D)$.
\end{defi}

We remark that this problem formulation is very general. For example, defining the joint distortion measure $d(X,Y,\hat{Z})$ as $d_1(F(X,Y),\hat{Z})$ enables us to consider the problem of lossy reconstruction of a function of the sources as a special case. Though we only consider a single distortion measure in this paper, it is straightforward to extend the results that we present here for the case of multiple distortion criteria. This implies that the problem of reconstructing the sources subject to two independent distortion criteria (the Berger-Tung problem \cite{berger-tung}) can be subsumed in this formulation with multiple distortion criteria. The Slepian-Wolf \cite{slepian-wolf} problem, the Wyner-Ziv problem \cite{wyner-ziv}, the Yeung-Berger problem \cite{yeung-berger} and the problem of coding with partial side information \cite{wyner75,ahlswede-korner} can also be subsumed by this formulation since they all are special cases of the Berger-Tung problem. The problem of remote distributed source coding  \cite{flynn-gray,viswanathan-berger-old}, where the encoders observe the sources through noisy channels, can also be subsumed in this formulation using the techniques of \cite{dobrushin-tsybakov,witsenhausen}. We shall see that our coding theorem has implications on the tightness of the Berger-Tung inner bound \cite{berger-tung}. The two-user function computation problem of lossy reconstruction of $Z = F(X,Y)$ can also be viewed as a special case of three-user Berger-Tung problem of encoding the correlated sources $(X,Y,Z)$ with three independent distortion criteria, where the rate of the third encoder is set to zero and the distortions of the first two sources are set to their maximum values.  We shall see in Section \ref{subsec:lossyxorex} that for this problem, our rate region indeed yields points outside the Berger-Tung rate region thus demonstrating that the Berger-Tung inner bound is not tight for the case of three or more sources.

An achievable rate region for the problem defined in Definitions \ref{defi:probdefi1} and \ref{defi:probdefi2} can be obtained based on the Berger-Tung coding scheme \cite{berger-tung} as follows. Let $\mc{P}$ denote the family of pair of conditional probabilities $(P_{U|X},P_{V|Y})$ defined on $\mc{X} \times \mc{U}$ and $\mc{Y} \times \mc{V}$, where $U$ and $V$ are finite sets. For any $(P_{U|X},P_{V|Y}) \in \mc{P}$, let the induced joint distribution be $P_{XYUV} = P_{XY} P_{U|X} P_{V|Y}$. $U,V$ play the role of auxiliary random variables. Define $G \colon \mc{U} \times \mc{V} \rightarrow \hat{\mc{Z}}$ as that function of $U,V$ that gives the optimal reconstruction $\hat{Z}$ with respect to the distortion measure $d(\cdot,\cdot,\cdot)$. With these definitions, an achievable rate region for this problem is presented below.

\begin{fact} \label{lemma:BTschemelemma}
For a given source $(X,Y)$ and distortion $d(\cdot,\cdot,\cdot)$ define the region $\mc{RD}_{BT}$ as
\begin{align}
\nonumber \mc{RD}_{BT} \triangleq \bigcup_{(P_{U|X},P_{V|Y}) \in
  \mc{P}} \left\{ R_1 \geq I(X;U|V), \, R_2 \geq I(Y;V|U), \, R_1 +
R_2 \geq I(XY;UV), \phantom{\sum_{i=1}^n \frac{1}{n}} \right. \\
\left. \phantom{\sum_{i=1}^n \frac{1}{n}} D \geq \Bbb{E}d(X,Y,G(U,V)) \right\} \phantom{aaaa}
\end{align}
Then any $(R_1,R_2,D) \in \mc{RD}_{BT}^{*}$ is achievable where $^*$ denotes convex closure\footnote{The cardinalities of $U$ and $V$ can be bounded using Caratheodary theorem \cite{csiszarbook}.}.
\end{fact}

\begin{proof}[\textbf{Proof}:]
Follows from the analysis of the Berger-Tung problem \cite{berger-tung} in a straightforward way.
\end{proof}

\section{Groups - An Introduction} \label{sec:groupintro}

In this section, we present an overview of some properties of groups that are used later. We refer the reader to \cite{dummit-foote} for more details. It is assumed that the reader has some basic familiarity with the concept of groups. We shall deal exclusively with abelian groups and hence the additive notation will be used for the group operation. The group operation of the group $G$ is denoted by ${+}_G$. Similarly, the identity element of group $G$ is denoted by $e_G$. The additive inverse of $a \in G$ is denoted by $-a$. The subscripts are omitted when the group in question is clear from the context. A subset $H$ of a group $G$ is called a subgroup if $H$ is a group by itself under the same group operation $+_G$. This is denoted by $H < G$. The direct sum of two groups $G_1$ and $G_2$ is denoted by $G_1 \oplus G_2$. The direct sum of a group $G$ with itself $n$ times is denoted by $G^n$.

An important tool in studying the structure of groups is the concept of group homomorphisms.
\begin{defi}
Let $G,H$ be groups. A function $\phi \colon G \rightarrow H$ is called a homomorphism if for any $a,b \in G$
\begin{equation}
\phi(a +_G b) = \phi(a) +_H \phi(b).
\end{equation}
A bijective homomorphism is called an isomorphism. If $G$ and $H$ are isomorphic, it is denoted as $G \cong H$.
\end{defi}
A homomorphism $\phi(\cdot)$ has the following properties: $\phi(e_G) = e_H$ and $\phi(-a) = -\phi(a)$. The kernel $\ker(\phi)$ of a homomorphism is defined as $ \ker(\phi) \triangleq \{x \in G \colon \phi(x) = e_H\}$. An important property of homomorphisms is that they preserve the subgroup structure. Let $\phi \colon G \rightarrow H$ be a homomorphism. Let $A<G$ and $B<H$. Then $\phi^{-1}(B) < G$ and $\phi(A) < B$. In particular, taking $B = \{e_H\}$, we get that $\ker(\phi) < G$.

One can define a congruence result analogous to number theory using subgroups of a group. Let $H < G$. Consider the set $Ha = \{h+a: h \in H \}$. The members of this set form an equivalence class called the right coset of $H$ in $G$ with $a$ as the coset leader. The left coset of $H$ in $G$ is similarly defined. Since we deal exclusively with abelian groups, we shall not distinguish cosets as being left or right. All cosets are of the same size as $H$ and two different cosets are either distinct or identical. Thus, the set of all distinct cosets of $H$ in $G$ form a partition of $G$. These properties shall be used in our coding scheme.

It is known that a finite cyclic group of order $n$ is isomorphic to the group $\ring{n}$ which is the set of integers $\{0, \dots, n-1\}$ with the group operation as addition modulo-$n$. A cyclic group whose order is the power of a prime is called a primary cyclic group. The following fact demonstrates the role of primary cyclic groups as the building blocks of all finite abelian groups.

\begin{fact} \label{fact:Gdecompfact}
Let $G$ be a finite abelian group of order $n > 1$ and let the unique factorization of $n$ into distinct prime powers be $n = \prod_{i=1}^k p_i^{e_i}$. Then,
\begin{equation} \label{eq:Gdecompeq1}
G \cong A_1 \oplus A_2 \dots \oplus A_k \quad \mbox{where } |A_i| = p_i^{e_i}
\end{equation}
Further, for each $A_i, 1 \leq i \leq k$ with $|A_i| = p_i^{e_i}$, we have
\begin{equation} \label{eq:Gdecompeq2}
A_i \cong \ring{p_i^{h_1}} \oplus \ring{p_i^{h_2}} \dots \oplus \ring{p_i^{h_t}}
\end{equation}
where $h_1 \geq h_2 \dots \geq h_t$ and $\sum_{j=1}^{t} h_j = e_i$. This decomposition of $A_i$ into direct sum of primary cyclic groups is called the invariant factor decomposition of $A_i$. Putting equations (\ref{eq:Gdecompeq1}) and (\ref{eq:Gdecompeq2}) together, we get a decomposition of an arbitrary abelian group $G$ into a direct sum of possibly repeated primary cyclic groups. Further, this decomposition of $G$ is unique,i.e., if $G \cong B_1 \oplus B_2 \dots B_m$ with $|B_i| = p_i^{e_i}$ for all $i$, then $B_i \cong A_i$ and $B_i$ and $A_i$ have the same invariant factors.
\end{fact}
\begin{proof}[\textbf{Proof}:]
See \cite{dummit-foote}, Section $5.2$, Theorem $5$.
\end{proof}

For example, Fact \ref{fact:Gdecompfact} implies that any abelian group of order $8$ is isomorphic to either $\ring{8}$ or $\ring{4} \oplus \ring{2}$ or to $\ring{2} \oplus \ring{2} \oplus \ring{2}$ where $\oplus$ denotes the direct sum of groups. Thus, we first consider the coding theorems only for the primary cyclic groups $\ringr{}$. Results obtained for such groups are then extended to hold for arbitrary abelian groups through this decomposition. Suppose $G$ has a decomposition $G \cong \ring{p_1^{e_1}} \oplus \dots \oplus \ring{p_r^{e_r}}$ where $p_1 \geq \dots \geq p_r$ are primes. A random variable $X$ taking values in $G$ can be thought of as a vector valued random variable $X = (X_1,\dots,X_r)$ with $X_i$ taking values in the cyclic group $\ring{p_i^{e_i}}, 1 \leq i \leq r$. $X_i$ are called the digits of $X$.

We now present some properties of primary cyclic groups that we shall use in our proofs. The group $\ring{m}$ is a commutative ring with the addition operation being addition modulo-$m$ and the multiplication operation being multiplication modulo-$m$. This multiplicative structure is also exploited in the proofs. The group operation in $\ring{m}^n$ is denoted by $u_1^n + u_2^n$. Addition of $u_1^n$ with itself $k$ times is denoted by $k u_1^n$. The multiplication operation between elements $x$ and $y$ of the underlying ring $\ring{m}$ is denoted by $xy$. We shall say that $x \in \ring{m}$ is invertible if there exists $y \in \ring{m}$ such that $xy = 1$ where $1$ is the multiplicative identity of $\ring{m}$. The multiplicative inverse of $x \in \ring{m}$, if it exists, is denoted by $x^{-1}$. The additive inverse of $u_1^n \in \ring{m}^n$ which always exists is denoted by $-u_1^n$. The group operation in the group $\ring{m}$ is often explicitly denoted by $\oplus_m$.

We shall build our codebooks as kernels of homomorphisms from $\ringr{n}$ to $\ringr{k}$. Justification for restricting the domain of our homomorphisms to $\ringr{n}$ comes from the decomposition result of Fact \ref{fact:Gdecompfact}. The reason for restricting the image of the homomorphisms to $\ringr{k}$ shall be made clear later on (see the proof of Lemma \ref{lemma:uniflem}). We need the following lemma on the structure of homomorphisms from $\ringr{n}$ to $\ringr{k}$.
\begin{fact}
Let $\mbox{Hom}(\ringr{n},\ringr{k})$ be the set of all homomorphisms from the group $\ringr{n}$ to $\ringr{k}$ and $M(k,n,\ringr{})$ be the set of all $k \times n$ matrices whose elements take values from the group $\ringr{}$. Then, there exists a bijection between $\mbox{Hom}(\ringr{n},\ringr{k})$ and $M(k,n,\ringr{})$ given by the invertible mapping $f \colon \mbox{Hom}(\ringr{n},\ringr{k}) \rightarrow M(k,n,\ringr{})$ defined as $f(\phi) = \Phi$ such that $\phi(x^n) = \Phi \cdot x^n$ for all $x^n \in \ringr{n}$. Here, the multiplication and addition operations involved in the matrix multiplication are carried out modulo-$p^r$.
\end{fact}

\begin{proof}[\textbf{Proof}:]
See \cite{kurosh}, Section VI.
\end{proof}

\section{Motivation of the Coding Scheme} \label{sec:motivation}

In this section, we present a sketch of the ideas involved in our coding scheme by demonstrating them for the simple case when the sources are binary. The emphasis in this section is on providing an overview of the main ideas and the exposition is kept informal. Formal definitions and theorems follow in subsequent sections. We first review the linear coding strategy of \cite{korner-marton} to reconstruct losslessly the modulo-$2$ sum of $Z = X \oplus_2 Y$ of the binary sources $X$ and $Y$. We then demonstrate that the Slepian-Wolf problem can be solved by a similar coding strategy. We generalize this coding strategy for the case when the fidelity criterion is such that the decoder needs to losslessly reconstruct a function $F(X,Y)$ of the sources. This shall motivate the problem of building ``good'' channel codes over abelian groups. We then turn our attention to the lossy version of the problem where the sources $X$ and $Y$ are quantized to $U$ and $V$ respectively first. For this purpose, we need to build ``good'' source codes over abelian groups. Then, encoding is done in such a way that the decoder can reconstruct $G(U,V)$ which is the optimal reconstruction of the sources with respect to the fidelity criterion $d(\cdot,\cdot,\cdot)$ given $U,V$. This shall necessitate the need for ``good'' nested group codes where the coarse code is a good channel code and the fine code is a good source code. These concepts shall be made precise later on in Sections \ref{sec:definitions} and \ref{sec:codingthm}.

\subsection{Lossless Reconstruction of the Modulo-2 Sum of the Sources} \label{subsec:kmintro}
This problem was studied in \cite{korner-marton} where an ingenious coding scheme involving linear codes was presented. This coding scheme can be understood as follows. It is well known \cite{wyner74} that linear codes can be used to losslessly compress a source down to its entropy. Formally, for any binary memoryless source $Z$ with distribution $p_Z(z)$ and any $\epsilon > 0$, there exists a $k \times n$ binary matrix $A$ with $\frac{k}{n} \leq H(Z) + \epsilon$ and a function $\psi$ such that
\begin{equation} \label{eq:Amatdef}
P(\psi(Az^n) \neq z^n) < \epsilon
\end{equation}
for all sufficiently large $n$. Let $Z = X \oplus_2 Y$ be the modulo-2 sum of the binary sources $X$ and $Y$. Let the matrix $A$ satisfy equation (\ref{eq:Amatdef}). The encoders of $X$ and $Y$ transmit $s_1 = Ax^n$ and $s_2 = Ay^n$ respectively at rates $(H(Z),H(Z))$. The decoder, upon receiving $s_1$ and $s_2$, computes $\psi(s_1 \oplus_2 s_2) = \psi(Ax^n \oplus_2 Ay^n) = \psi(Az^n)$. Since the $A$ matrix was chosen in accordance with equation (\ref{eq:Amatdef}), the decoder output equals $z^n$ with high probability. Thus, the rate pair $(H(Z),H(Z))$ is achievable. If the source statistics is such that $H(Z) > H(X)$, then clearly it is better to compress $X$ at a rate $H(X)$. Thus, the Korner-Marton coding scheme achieves the rate pair $(R_1,R_2)$ with $R_1 \geq \min \{H(X),H(Z)\}$ and $R_2 \geq \min\{H(Y),H(Z)\}$. This coding strategy shall be referred to as the Korner-Marton coding scheme from here on.

The crucial part played by linear codes in this coding scheme is noteworthy. Had there been a centralized encoder with access to $x^n$ and $y^n$, the coding scheme would be to compute $z^n = x^n \oplus_2 y^n$ first and then compress it using any method known to achieve the entropy bound. Because the encoding is linear, it enables the decoder to use the \emph{distributive} nature of the linear code over the modulo-2 operation to compute $s_1 \oplus_2 s_2 = Az^n$. Thus, from the decoder's perspective, there is no distinction between this distributed coding scheme and a centralized scheme involving a linear code. Also, in contrast to the usual norm in information theory, there is no other known coding scheme that approaches the performance of this linear coding scheme.

More generally, in the case of a prime $q$, a sum rate of $2H(X \oplus_q Y)$ can be achieved \cite{han-kobayashi} for the reconstruction of the sum of the two $q$-ary sources $Z = X \oplus_q Y$ in any prime field $\ring{q}{}$. Abstractly, the Korner-Marton scheme can be thought of as a structured coding scheme with codes built over groups that enable the decoder to reconstruct the group operation losslessly. It turns out that extending the scheme would involve building ``good'' channel codes over arbitrary abelian groups. It is known (see Fact \ref{fact:Gdecompfact}) that primary cyclic groups $\ring{p^r}$ are the building blocks of all abelian groups and hence it suffices to build ``good'' channel codes over the cyclic groups $\ring{p^r}$.

\subsection{Lossless Reconstruction of the Sources}
\label{subsec:swintro} The classical result of Slepian and Wolf \cite{slepian-wolf} states that it is possible to reconstruct the sources $X$ and $Y$ noiselessly at the decoder with a sum rate of $R_1 + R_2 = H(X,Y)$. As was shown in \cite{csiszar82}, the Slepian-Wolf bound is achievable using linear codes. Here, we present an interpretation of this linear coding scheme and connect it to the one in the previous subsection. We begin by making the observation that reconstructing the function $Z = (X,Y)$ for binary sources can be thought of as reconstructing a linear function in the field $\ring{2} \oplus \ring{2}$. This equivalence is demonstrated below. Let the elements of $\ring{2} \oplus \ring{2}$ be $\{00,01,10,11\}$. Denote the addition operation of $\ring{2} \oplus \ring{2}$ by $\oplus_K$.

Define the mappings
\begin{align}
\tilde{X} &= \left\{ \begin{array}{cc} 00 & \mbox{if } X = 0 \\ 01 & \mbox{if }X = 1 \end{array} \right. \\
\tilde{Y} &= \left\{ \begin{array}{cc} 00 & \mbox{if } Y = 0 \\ 10 & \mbox{if }Y = 1 \end{array} \right.
\end{align}

Clearly, reconstructing $(X,Y)$ losslessly is equivalent to reconstructing the function $\tilde{Z} = \tilde{X} \oplus_K \tilde{Y}$ losslessly. The next observation is that elements in $\ring{2} \oplus \ring{2}$ can be represented as two dimensional vectors whose components are in $\ring{2}$. Further, addition in $\ring{2} \oplus \ring{2}$ is simply vector addition with the components of the vector added in $\ring{2}$. Let the first and second bits of $\tilde{X}$ be denoted by $\tilde{X}_1$ and $\tilde{X}_2$ respectively. The same notation holds for $\tilde{Y}$ and $\tilde{Z}$ as well. Then, we have the decomposition of the vector function $\tilde{Z}$ as $\tilde{Z}_i = \tilde{X}_i \oplus_2 \tilde{Y}_i$ for $i = 1,2$.

Encoding the vector function $\tilde{Z}$ directly using the Korner-Marton coding scheme would entail a sum rate of $R_1 + R_2 = \min\{H(X,Y),H(X)\} + \min\{H(X,Y),H(Y)\} = H(X) + H(Y)$ which is more than the sum rate dictated by the Slepian-Wolf bound. Instead, we encode the scalar components of the function $\tilde{Z}$ sequentially using the Korner-Marton scheme. Suppose the first digit plane $\tilde{Z}_1$ is encoded first. Assuming that it gets decoded correctly at the decoder, it is available as side information for the encoding of the second digit plane $\tilde{Z}_2$. Clearly, the Korner-Marton scheme can be used to encode the first digit plane $\tilde{Z}_1$. The rate pair $(R_{11},R_{21})$ achieved by the scheme is given by
\begin{align}
R_{11} &\geq \min \{ H(\tilde{Z}_1), H(\tilde{X}_1) \} = H(\tilde{X}_1) = 0 \\
R_{21} &\geq \min \{ H(\tilde{Z}_1), H(\tilde{Y}_1) \} = H(\tilde{Z}_1)
\end{align}

It is straightforward to extend the Korner-Marton coding scheme to the case where decoder has available to it some side information. Since $\tilde{Z}_1$ is available as side information at the decoder, the rates needed to encode the second digit plane $\tilde{Z}_2$ are
\begin{align}
R_{12} &\geq \min \{ H(\tilde{Z}_2 \mid \tilde{Z}_1), H(\tilde{X}_2
\mid \tilde{Z}_1) \} = H(\tilde{Z}_2 \mid \tilde{Z}_1) \\
R_{22} &\geq \min \{ H(\tilde{Z}_2 \mid \tilde{Z}_1), H(\tilde{Y}_2 \mid \tilde{Z}_1) \} = H(\tilde{Y}_2 \mid \tilde{Z}_1) = 0
\end{align}

Thus, the overall rate pair needed to reconstruct the sources losslessly is
\begin{align}
R_1 &= R_{11} + R_{12} \geq  H(\tilde{Z}_2 \mid \tilde{Z}_1) = H(\tilde{X}_2 \mid \tilde{Y}_1) \\
R_2 &= R_{21} + R_{22} \geq H(\tilde{Z}_1) = H(\tilde{Y}_1).
\end{align}
The sum rate for this scheme is $R_1 + R_2 = H(\tilde{X}_2,\tilde{Y}_1) = H(X,Y)$ thus equaling the Slepian-Wolf bound.

\subsection{Lossless Reconstruction of an Arbitrary Function $F(X,Y)$ }
\label{subsec:arbitfuncintro}

While there are more straightforward ways of achieving the Slepian-Wolf bound than the method outlined in Section \ref{subsec:swintro}, our encoding scheme has the advantage of putting the Korner-Marton coding scheme and the Slepian-Wolf coding scheme under the same framework. The ideas used in these two examples can be abstracted and generalized for the problem when the decoder needs to losslessly reconstruct some function $F(X,Y)$ in order to satisfy the fidelity criterion.

Let us assume that the cardinality of $X$ and $Y$ are respectively $\alpha$ and $\beta$. The steps involved in such an encoding scheme can be described as follows. We first represent the function as equivalent to the group operation in some abelian group $A$. This is referred to as ``embedding'' the function in $A$. This abelian group is then decomposed into its constituent cyclic groups and the embedded function is sequentially encoded using the Korner-Marton scheme outlined in Section \ref{subsec:kmintro}. Encoding is done keeping in mind that, to decode a digit, the decoder has as available side information all previously decoded digits.

It suffices to restrict attention to abelian groups $A$ such that $|\mc{Z}| \leq |A| \leq \alpha \beta$. Clearly, if the function $F_1(X,Y) \triangleq (X,Y)$ can be embedded in a certain abelian group, then any function $F(X,Y)$ can be reconstructed in that abelian group. This is because the decoder can proceed by reconstructing the sources $(X,Y)$ and then computing the function $F(X,Y)$. It can be shown (see Appendix \ref{sec:Tnonempty}) that the function $F_1(X,Y) \triangleq (X,Y)$ can be reconstructed in the group $\ring{\alpha} \oplus \ring{\beta}$ which is of size $\alpha \beta$. Clearly, $|A| \geq |\mc{Z}|$ is a necessary condition for the reconstruction of $Z = F(X,Y)$.

\subsection{Lossy Reconstruction} \label{subsec:lossyreconintro}

We now turn our attention to the case when the decoder wishes to obtain a reconstruction $\hat{Z}$ with respect to a fidelity criterion. The coding strategy is as follows: Quantize the sources $X$ and $Y$ to auxiliary variables $U$ and $V$. Given the quantized sources $U$ and $V$, let $G(U,V)$ be the optimal reconstruction with respect to the distortion measure $d(\cdot,\cdot,\cdot)$. Reconstruct the function $G(U,V)$ losslessly using the coding scheme outlined in Section \ref{subsec:arbitfuncintro}.

We shall use nested group codes to effect this quantization. Nested group codes arise naturally in the area of distributed source coding and require that the fine code be a ``good'' source code and the coarse code be a ``good'' channel code for appropriate notions of goodness. We have already seen that to effect lossless compression, the channel code operates at the digit level. It follows then that we must use a series of nested group codes, one for each digit, over appropriate cyclic groups. For instance, if the first digit of $G(U,V)$ is over the cyclic group $\ring{p_1^{e_1}}{}$, then we need nested group codes over $\ring{p_1^{e_1}}{}$ that encode the sources $X$ and $Y$ to $\tilde{U}_1$ and $\tilde{V}_1$ respectively. The quantization operation is also carried out sequentially, i.e., the digits $\tilde{U}_2$ and $\tilde{V}_2$ are encoded given the knowledge that either $\tilde{Z}_1$ or $(\tilde{U}_1,\tilde{V}_1)$ is available at the decoder and so on. The existence of ``good'' nested group codes over arbitrary cyclic groups is shown later.

The steps involved in the overall coding scheme can be detailed as follows:
\begin{itemize}
\item Let $U,V$ be discrete random variables over the alphabet $\mc{U},\mc{V}$ respectively. Further suppose that $|\mc{U}| = \alpha, |\mc{V}| =
    \beta$. Choose the joint density $P_{X,Y,U,V} = P_{X,Y}P_{U|X}P_{V|Y}$ satisfying the Markov chain $U-X-Y-V$.
\item Let $G(U,V)$ be the optimal reconstruction function with respect to $d(\cdot,\cdot,\cdot)$ given $U,V$. \item Embed the function $G(U,V)$ in an abelian group $A$,  $|\mc{G}| \leq |A| \leq \alpha \beta$. \item Decompose $G(U,V)$ into its constituent digit planes. Fix the order in which the digit planes are to be sequentially encoded. \item Suppose the $b^{\mbox{th}}$ digit plane is the cyclic group $\ring{p_b^{e_b}}{}$. Quantize the sources $(X^n,Y^n)$ into digits
    $(\tilde{U}_b, \tilde{V}_b)$ using the digits already available at the decoder as side information. The details of the quantization procedure
    are detailed later.
\item Encode $\tilde{Z}_b = \tilde{U}_b \oplus_{p_b^{e_b}} \tilde{V}_b$ using group codes.
\end{itemize}

\section{Definitions} \label{sec:definitions}

When a random variable $X$ takes value over the group $\ringr{}$, we need to ensure that it doesn't just take values in some proper subgroup of $\ringr{}$. This leads us to the concept of a non-redundant distribution over a group.

\begin{defi} \label{defi:nrdist}
A random variable $X$ with $\mc{X} = \ringr{}$ or its distribution $P_X$ is said to be non-redundant if $P_X(x) > 0$ for at least one symbol $x \in \ringr{} \backslash p \ringr{}$.
\end{defi}
It follows from this definition that $x^n \in \typset{}{X}$ contains at least one $x \in \ringr{} \backslash p \ringr{}$ if $X$ is non-redundant. Such sequences are called non-redundant sequences. A redundant random variable taking values over $\ringr{}$ can be made non-redundant by a suitable relabeling of the symbols. Also, note that a redundant random variable over $\ringr{}$ is non-redundant when viewed as taking values over $\ring{p^{r-i}}$ for some $0 < i \leq r$. Our coding scheme involves good nested group codes for source and channel coding and the notion of embedding the optimal reconstruction function in a suitable abelian group. These concepts are made precise in the following series of definitions.

\begin{defi} \label{defi:embeddingdefi}
A bivariate function $G \colon \mc{U} \times \mc{V} \rightarrow \mc{G}$ is said to be embeddable in an abelian group $A$ with respect to the distribution $p_{UV}(u,v)$ on $\mc{U} \times \mc{V}$ if there exists injective functions $S_U^{(A)} \colon \mc{U} \rightarrow A, S_V^{(A)} \colon \mc{V} \rightarrow A$ and a surjective function $S_G^{(A)} \colon A \rightarrow \mc{G}$ such that
\begin{equation}
S_G^{(A)} (S_U^{(A)}(u) +_A S_V^{(A)}(v)) = G(u,v) \quad \forall (u,v) \in \mc{U} \times \mc{V} \mbox{ with } p_{UV}(u,v) > 0
\end{equation}
If $G(U,V)$ is indeed embeddable in the abelian group $A$, it is denoted as $G(U,V) \subset A$ with respect to the distribution $p_{UV}(u,v)$. Define the mapped random variables $\bar{U} = S_U^{(A)}(U)$ and $\bar{V} = S_V^{(A)}(V)$. Their dependence on $A$ is suppressed and the group in question will be clear from the context.
\end{defi}

Suppose the function $G(U,V) \subset A$ with respect to $p_{UV}$. We encode the function $G(U,V)$ sequentially by treating the sources as vector valued over the cyclic groups whose direct sum is isomorphic to $A$. This alternative representation of the sources is made precise in the following definition.
\begin{defi} \label{defi:decompdefi}
Suppose the function $G(U,V) \subset A$ with respect to $p_{UV}$. Let $A$ be isomorphic to $\oplus_{i=1}^k \ring{p_i^{e_i}}$ where $p_1 \leq \dots \leq p_k$ are primes and $e_i$ are positive integers. Then, it follows from Fact \ref{fact:Gdecompfact} that there exists a bijection $S_A \colon A \rightarrow \ring{p_1^{e_1}} \times \dots \ring{p_k^{e_k}}$. Let $\tilde{U} = S_A(\bar{U}), \tilde{V} = S_A(\bar{V})$. Let $\tilde{U} = ( \tilde{U}_1, \dots, \tilde{U}_k)$ be the vector representation of $\tilde{U}$. The random variables $\tilde{U}_i$ are called the digits of $\tilde{U}$. A similar decomposition holds for $\tilde{V}$. Define $\tilde{Z} = (\tilde{Z}_1, \dots, \tilde{Z}_k)$ where $\tilde{Z}_i \triangleq \tilde{U}_i \oplus_{p_i^{e_i}} \tilde{V}_i$. It follows that $S_A^{-1}(\tilde{Z}) = \bar{U} +_A \bar{V}$.
\end{defi}

Our encoding operation proceeds thus: we reconstruct the function $G(U,V)$ by first embedding it in some abelian group $A$ and then reconstructing $\bar{U} +_A \bar{V}$ which we accomplish sequentially by reconstructing $\tilde{U}_i \oplus_{p_i^{e_i}} \tilde{V}_i$ one digit at a time. While reconstructing the $i$th digit, the decoder has as side information the previously reconstructed $(i-1)$ digits. This digit decomposition approach requires that we build codes over the primary cyclic groups $\ringr{}$ which are ``good'' for various coding purposes. We define the concepts of group codes and what it means for group codes to be ``good'' in the following series of definitions.

\begin{defi} \label{defi:groupcodesdefi}
Let $A$ be a finite abelian group. A group code $\mc{C}$ of blocklength $n$ over the group $A$ is a subset of $A^n$ which is closed under the group addition operation, i.e., $\mc{C} \subset A^n$ is such that if $c_1^n, c_2^n \in \mc{C}$, then so does $c_1^n +_{A^n} c_2^n$.
\end{defi}
Recall that the kernel $\ker(\phi)$ of a homomorphism $\phi \colon A^n \rightarrow A^k$ is a subgroup of $A^n$. We use this fact to build group codes. As mentioned earlier, we build codes over the primary cyclic group $\ringr{}$. In this case, every group code $\mc{C} \subset \ringr{n}$ has associated with it a $k \times n$ matrix $H$ with entries in $\ringr{}$ which completely defines the group code as
\begin{equation}
\mc{C} \triangleq \{ x^n \in \ringr{n} \colon Hx^n = 0^k \}.
\end{equation}
Here, the multiplication and addition are carried out modulo-$p^r$. $H$ is called the parity-check matrix of the code $\mc{C}$. We employ nested group codes in our coding scheme. In distributed source coding problems, we often need one of the components of a nested code to be a good source code while the other one to be a good channel code. We shall now define nested group codes and the notions of ``goodness'' used to classify a group code as a good source or channel code.

\begin{defi} \label{defi:nestedgroupcodesdefi}
A nested group code $(\mc{C}_1, \mc{C}_2)$ is a pair of group codes such that every codeword in the codebook $\mc{C}_2$ is also a codeword in $\mc{C}_1$, i.e., $\mc{C}_2 < \mc{C}_1$. Their associated parity check matrices are the $k_1 \times n$ matrix $H_1$ and the $k_2 \times n$ matrix $H_2$. They are related to each other as $H_1 = J \cdot H_2$ for some $k_1 \times k_2$ matrix $J$. One way to enforce this relation between $H_1$ and $H_2$ would be to let
\begin{equation}
H_2 = \left[ \begin{array}{c} H_1 \\ \Delta H \end{array} \right]
\end{equation}
where $\Delta H$ is a $(k_2-k_1) \times n$ matrix over $\ringr{}$.
\end{defi}
The code $\mc{C}_1$ is called the fine group code while $\mc{C}_2$ is called the coarse group code. When nested group codes are used in distributed source coding, typically the coset leaders of $\mc{C}_2$ in $\mc{C}_1$ are employed as codewords. In such a case, the rate of the nested group code would be $n^{-1}(k_2 - k_1) \log p^r$ bits.

We define the notion of ``goodness'' associated with a group code below. To be precise, these notions are defined for a family of group codes indexed by the blocklength $n$. However, for the sake of notational convenience, this indexing is not made explicit.

\begin{defi} \label{defi:goodsrccodedefi}
Let $P_{XU}$ be a distribution over $\mc{X} \times \mc{U}$ such that the marginal $P_U$ is a non-redundant distribution over $\ringr{}$ for some prime power $p^r$. For a given group code $\mc{C}$ over $\mc{U}$ and a given $\epsilon > 0$, let the set $A_{\epsilon}(\mc{C})$ be defined as
\begin{equation} \label{eq:Asetdefi}
A_{\epsilon}(\mc{C}) \triangleq \{ x^n \colon \exists u^n \in \mc{C} \mbox{ such that } (x^n,u^n) \in A_{\epsilon}^{(n)}(X,U) \}.
\end{equation}
The group code $\mc{C}$ over $\mc{U}$ is called a good source code for the triple $(\mc{X},\mc{U}, P_{XU})$ if we have $\forall \epsilon > 0$,
\begin{equation} \label{eq:goodlinsrceq}
P_X^n(A_{\epsilon}(\mc{C})) \geq 1-\epsilon
\end{equation}
for all sufficiently large $n$.
\end{defi}

Note that, a group code which is a good source code in this sense may not be a good source code in the usual Shannon sense. Rather, such a group code contains a subset which is a good source code in the Shannon sense for the source $P_X$ with forward test channel $P_{U|X}$.

\begin{defi} \label{defi:goodchcodedefi}
Let $P_{ZS}$ be a distribution over $\mc{Z} \times \mc{S}$ such that the marginal $P_Z$ is a non-redundant distribution over $\ringr{}$ for some prime power $p^r$. For a given group code $\mc{C}$ over $\mc{Z}$ and a given $\epsilon > 0$, define the set $B_{\epsilon}(\mc{C})$ as follows:
\begin{equation} \label{eq:Bsetdefi}
B_{\epsilon}(\mc{C}) \triangleq \{ (z^n,s^n) \colon \exists \tilde{z}^n \mbox{ such that } (\tilde{z}^n,s^n) \in A_{\epsilon}^{(n)}(Z,S) \mbox{ and } H\tilde{z}^n = Hz^n \}.
\end{equation}
Here, $H$ is the $k(n) \times n$ parity check matrix associated with the group code $\mc{C}$. The group code $\mc{C}$ is called a good channel code for the triple $(\mc{Z},\mc{S},P_{ZS})$ if we have $\forall \epsilon > 0$,
\begin{equation} \label{eq:goodlinchanneleq}
P^n_{ZS}(B_{\epsilon}(\mc{C})) \leq \epsilon
\end{equation}
for all sufficiently large $n$. Associated with such a good group channel code would be a decoding function $\psi:\ring{p^r}^k \times \mc{S}^n \rightarrow \ring{p^r}^n$ such that
\begin{equation}
P(\psi(Hz^n,s^n) = z^n) \geq 1 - \epsilon.
\end{equation}
\end{defi}

Note that, as before, a group code which is a good channel code in this sense may not a good channel code in the usual Shannon sense. Rather, every coset of such a group code contains a subset which is a good channel code in the Shannon sense for the channel $P_{S|Z}$ with input distribution $P_Z$. This interpretation is valid only when $S$ is a non-trivial random variable.

\begin{lemma} \label{lemma:goodchcodelemma}
For any triple $(\mc{Z},\mc{S},P_{ZS})$ of two finite sets and a distribution, with $|\mc{Z}| = p^r$ a prime power and $P_Z$ non-redundant, there exists a sequence of group codes $\mc{C}$ that is a good channel code for the triple $(\mc{Z},\mc{S},P_{ZS})$ such that the dimensions of their associated $k(n) \times n$ parity check matrices satisfy
\begin{equation} \label{eq:goodlinchannellimiteq}
\lim_{n \rightarrow \infty} \frac{k(n)}{n} \log p^r = \max_{0 \leq i <
  r} \left( \frac{r}{r-i} \right) (H(Z|S) - H([Z]_i | S))
\end{equation}
where $[Z]_i$ is a random variable taking values over the set of all distinct cosets of $p^i \ringr{}$ in $\ringr{}$. For example, if $\mc{Z} = \ring{8}$, then $[Z]_2$ is a $4$-ary random variable with symbol probabilities $(p_Z(0) + p_Z(4)), (p_Z(1) + p_Z(5)), (p_Z(2) + p_Z(6))$ and $(p_Z(3) + p_Z(7))$.
\end{lemma}
\begin{proof}[\textbf{Proof}:]
See Appendix \ref{sec:goodchproof}.
\end{proof}

Note that $[Z]_0$ is a constant and $[Z]_r = Z$. When building codes over groups, each proper subgroup of the group contributes a term to the maximization in equation (\ref{eq:goodlinchannellimiteq}). Since the smaller the right hand side of equation (\ref{eq:goodlinchannellimiteq}), the better the channel code is, we incur a penalty by building codes over groups with large number of subgroups.

\begin{lemma} \label{lemma:goodsrccodelemma}
For any triple $(\mc{X},\mc{U},P_{XU})$ of two finite sets and a distribution, with $|\mc{U}| = p^r$ a prime power and $P_U$ non-redundant, there exists a sequence of group codes $\mc{C}$ that is a good source code for the triple $(\mc{X},\mc{U},P_{XU})$ such that the dimensions of their associated $k(n) \times n$ parity check matrices satisfy
\begin{equation} \label{eq:goodlinsrclimiteq}
\lim_{n \rightarrow \infty} \frac{k(n)}{n} \log p^r = \min (H(U|X), r|H(U|X) - \log p^{r-1}|^{+})
\end{equation}
where $|x|^{+} = \max(x,0)$.
\end{lemma}
\begin{proof}[\textbf{Proof}:]
See Appendix \ref{sec:goodsrcproof}.
\end{proof}

Putting $r=1$ in equations (\ref{eq:goodlinchannellimiteq}) and (\ref{eq:goodlinsrclimiteq}), we get the performance obtainable while using linear codes built over Galois fields.

\begin{lemma} \label{lemma:goodnestedcodelemma}
Let $X,Y,S,U,V$ be five random variables where $U$ and $V$ take value over the group $\ringr{}$ for some prime power $p^r$. Let $Z = U \oplus_{p^r} V$. Let $U \rightarrow X \rightarrow Y \rightarrow V$ form a Markov chain, and let $S \rightarrow (X,Y) \rightarrow (U,V)$ form a Markov chain. From the Markov chains, it follows that $H(U|X) \leq H(Z|S), H(V|Y) \leq H(Z|S)$. Without loss of generality, let $H(U|X) \leq H(V|Y) \leq H(Z|S)$. Then, there exists a pair of nested group codes $(\mc{C}_{11}, \mc{C}_{2})$ and $(\mc{C}_{12},\mc{C}_2)$ such that
\begin{itemize}
\item $\mc{C}_{11}$ is a good group source code for the triple $(\mc{X},\mc{U},P_{XU})$ with
\begin{equation} \label{eq:nestedeq1}
\lim_{n \rightarrow \infty} \frac{k_{11}(n)}{n} \log p^r = \min(H(U|X), r|H(U|X) - \log p^{r-1}|^{+})
\end{equation}
\item $\mc{C}_{12}$ is a good group source code for the triple $(\mc{Y},\mc{V},P_{YV})$ with
\begin{equation} \label{eq:nestedeq2}
\lim_{n \rightarrow \infty} \frac{k_{12}(n)}{n} \log p^r = \min(H(V|Y), r|H(V|Y) - \log p^{r-1}|^{+})
\end{equation}
\item $\mc{C}_2$ is a good group channel code for the triple $(\mc{Z},\mc{S},P_{ZS})$ with
\begin{equation} \label{eq:nestedeq3}
\lim_{n \rightarrow \infty} \frac{k_2(n)}{n} \log p^r = \max_{0 \leq i < r} \left( \frac{r}{r-i} \right) (H(Z|S) - H([Z]_i|S))
\end{equation}
\end{itemize}
\end{lemma}

\begin{proof}[\textbf{Proof}:]
See Appendix \ref{sec:goodsrcchproof}
\end{proof}

Note that while choosing the codebooks $\mc{C}_{11}, \mc{C}_{12}$ and $\mc{C}_2$, the perturbation parameters $\epsilon$ in Definitions \ref{defi:goodsrccodedefi} and \ref{defi:goodchcodedefi} need to be chosen appropriately relative to each other so that the $n$-length sequences $(X^n,Y^n,S^n,U^n,V^n,Z^n)$ are jointly typical with high probability. Due to the Markov chains $U \rightarrow X \rightarrow Y \rightarrow V$ and $S \rightarrow (X,Y) \rightarrow (U,V)$, it follows from Markov lemma \cite{berger-lecture} that if $(X^n,Y^n,S^n)$ is generated according to $P_{XYS}$ and if $U^n$ is generated jointly typical with $X^n$ and $V^n$ is generated jointly typical with $Y^n$, then $(X^n,Y^n,S^n,U^n,V^n,Z^n)$ is jointly strongly typical (for an appropriate choice of $\epsilon$) with high probability.

\section{The Coding Theorem} \label{sec:codingthm}

We are given discrete random variables $X$ and $Y$ which are jointly distributed according to $P_{XY}$. Let $\mc{P}$ denote the family of pair of conditional probabilities $(P_{U|X},P_{V|Y})$ defined on $\mc{X} \times \mc{U}$ and $\mc{Y} \times \mc{V}$, where $\mc{U}$ and $\mc{V}$ are finite sets, $|\mc{U}| = \alpha, |\mc{V}| = \beta$. For any $(P_{U|X},P_{V|Y}) \in \mc{P}$, let the induced joint distribution be $P_{XYUV} = P_{XY} P_{U|X} P_{V|Y}$. $U,V$ play the role of auxiliary random variables. Define $G \colon \mc{U} \times \mc{V} \rightarrow \hat{\mc{Z}}$ as that function of $U,V$ that gives the optimal reconstruction $\hat{Z}$ with respect to the distortion measure $d(\cdot,\cdot,\cdot)$. Let $\mc{G}$ denote the image of $G(U,V)$. Let $\mc{T} = \{ A \colon A \mbox{ is abelian}, |\mc{G}| \leq |A| \leq \alpha \beta, \, G(U,V) \subset A \mbox{ with respect to } P_{UV}\}$. It is shown in Appendix \ref{sec:Tnonempty} that the set $\mc{T}$ is non-empty, i.e., there always exists an abelian group $A \in \mc{T}$ in which any function $G(U,V)$ can be embedded. For any $A \in \mc{T}$,  let $A$ be isomorphic to $\oplus_{i=1}^k \ring{p_i^{e_i}}$. Let $\tilde{U} = S_A(S_U^{(A)}(U))$ and $\tilde{V} = S_A(S_V^{(A)}(V))$ where the mappings are as defined in Definitions \ref{defi:embeddingdefi} and \ref{defi:decompdefi}. Define $\tilde{Z} = (\tilde{Z}_1,\dots,\tilde{Z}_k)$ where $\tilde{Z}_i = \tilde{U}_i \oplus \tilde{V}_i$ and the addition is done in the group to which the digits $\tilde{U}_i,\tilde{V}_i$ belong. Assume without loss of generality  that the digits $\tilde{U}_i,\tilde{V}_i,\tilde{Z}_i, 1 \leq i \leq k$ are all non-redundant. If they are not, they can be made so by suitable relabeling of the symbols. Recall the definition of $[Z]_i$ from Lemma \ref{lemma:goodchcodelemma}. The encoding operation of the $X$ and $Y$ encoders proceed in $k$ steps with each step producing one digit of $\tilde{U}$ and $\tilde{V}$ respectively. Let $\pi_A \colon \{1, \dots, k \} \rightarrow \{1, \dots, k\}$ be a permutation. The permutation $\pi_A$ can be thought of as determining the order in which the digits get encoded and decoded. Let the set $\Pi_A(b), 1 \leq b \leq k$ be defined as $\Pi_A(b) = \{ l \colon \pi_A(l) < b \}$. The set $\Pi_A(b)$ contains the indices of all the digits that get encoded before the $b$th stage. At the $b$th stage, let the digits $\tilde{U}_{\pi_A(b)}, \tilde{V}_{\pi_A(b)}$ take values over the group $\ring{p_b}^{r_b}$. With these definitions, an achievable rate region for the problem is presented below.

\begin{thm} \label{thm:mainthm}
For a given source $(X,Y)$, define the region $\mc{RD}_{in}$ as
\begin{align}
\mc{RD}_{in} \triangleq \bigcup_{\stackrel{(P_{U|X}, P_{V|Y}) \in
    \mc{P}}{A \in \mc{T}, \pi_A}} \left\{ (R_1,R_2,D) \colon R_1 \geq
\sum_{b=1}^{k} \min \left(R_{1b}^{(1)}, R_{1b}^{(2)} \right), R_2 \geq \sum_{b=1}^{k} \min \left(R_{2b}^{(1)}, R_{2b}^{(2)} \right)
\right. \\
\left. \phantom{R_1 \geq \sum_{b=1}^{k} \min \left(R_{1b}^{(1)},
    R_{1b}^{(2)} \right)}  D \geq \Bbb{E} d(X,Y, G(U,V)) \right\}
\end{align}
where
\begin{align}
\nonumber R_{1b}^{(1)} &> \left[ \max_{0 \leq i < r_b} \left( \frac{r_b}{r_b-i} \right) \left(H(\tilde{Z}_{\pi_A(b)} \mid \tilde{Z}_{\Pi_A(b)}) - H([\tilde{Z}_{\pi_A(b)}]_i | \tilde{Z}_{\Pi_A(b)}) \right) \right] \\ &- \left[ \min \left( H(\tilde{U}_{\pi_A(b)} \mid X, \tilde{U}_{\Pi_A(b)}), r_b(|H(\tilde{U}_{\pi_A(b)} \mid X, \tilde{U}_{\Pi_A(b)}) - \log p_b^{r_b-1}|^{+}) \right) \right]
\end{align}
and
\begin{align}
\nonumber R_{1b}^{(2)} &> \left[ \max_{0 \leq i < r_b} \left( \frac{r_b}{r_b-i} \right) \left(H(\tilde{U}_{\pi_A(b)} \mid \tilde{Z}_{\Pi_A(b)}) - H([\tilde{U}_{\pi_A(b)}]_i \mid \tilde{Z}_{\Pi_A(b)}) \right) \right] \\ &- \left[ \min \left( H(\tilde{U}_{\pi_A(b)} \mid X, \tilde{U}_{\Pi_A(b)}), r_b(|H(\tilde{U}_{\pi_A(b)} \mid X, \tilde{U}_{\Pi_A(b)}) - \log p_b^{r_b-1}|^{+}) \right) \right]
\end{align}
The quantities $R_{2b}^{(1)}$ and $R_{2b}^{(2)}$ are similarly defined with $(X,U)$ replaced by $(Y,V)$. Then any $(R_1,R_2,D) \in \mc{RD}_{in}^{*}$ is achievable where $^*$ denotes convex closure.
\end{thm}

\begin{proof}[\textbf{Proof}:]
Since the encoders don't communicate with each other, we impose the Markov chain $V-Y-X-U$ on the joint distribution $P_{XYUV}$. The family $\mc{P}$ contains all distributions that satisfy this Markov chain. Fix such a joint distribution. Fix $A \in \mc{T}$ and the permutation
 $\pi_A \colon \{1,\dots,k\} \rightarrow \{1,\dots,k\}$. The
encoding proceeds in $k$ stages with the $b$th stage encoding the digits $\tilde{U}_{\pi_A(b)}, \tilde{V}_{\pi_A(b)}$ in order to produce the digit $\tilde{Z}_{\pi_A(b)}$. For this, the decoder has side information $\tilde{Z}_{\Pi_A(b)}$.

Let $\tilde{U}_{\pi_A(b)}, \tilde{V}_{\pi_A(b)}$ take values over the group $\ringrb{}$. The encoders have two encoding options available at the $b$th stage. They can either encode the digits $\tilde{U}_{\pi_A(b)}$ and $\tilde{V}_{\pi_A(b)}$ directly or encode in such a way that the decoder is able to reconstruct $\tilde{Z}_{\pi_A(b)}$ directly. We present a coding scheme to achieve the latter first.

We shall use a pair of nested group codes $(\mc{C}_{11b}, \mc{C}_{2b})$ and $(\mc{C}_{12b}, \mc{C}_{2b})$ to encode $\tilde{Z}_{\pi_A(b)}$. Let the corresponding parity check matrices of these codes be $H_{11b}, H_{12b}$ and $H_{2b}$ respectively. Let the dimensionality of these matrices be $k_{11b} \times n$, $k_{12b} \times n$ and $k_{2b} \times n$ respectively. These codebooks are all over the group $\ringrb{}$. We need $C_{11b}$ to be a good source code for the triple $(\mc{X} \times \tilde{\mc{U}}_{\Pi_A(b)}, \tilde{\mc{U}}_{\pi_A(b)}, P_{X \tilde{U}_{\Pi_A(b)} \tilde{U}_{\pi_A(b)}})$, $\mc{C}_{12b}$ to be a good source code for the triple $(\mc{Y} \times \tilde{\mc{V}}_{\Pi_A(b)}, \tilde{\mc{V}}_{\pi_A(b)}, P_{Y \tilde{V}_{\Pi_A(b)} \tilde{V}_{\pi_A(b)}})$ and $\mc{C}_{2b}$ to be a good channel code for the triple $(\tilde{\mc{Z}}_{\pi_A(b)}, \tilde{\mc{Z}}_{\Pi_A(b)}, P_{\tilde{Z}_{\pi_A(b)} \tilde{Z}_{\Pi_A(b)}})$.

The encoding scheme used by the $X$-encoder to encode the $b$th digit, $1 \leq b \leq k$ is detailed below. The $X$-encoder looks for a typical sequence $\tilde{U}_{\pi_A(b)}^n \in \mc{C}_{11b}$ such that it is jointly typical with the source sequence $X^n$ and the previous encoder output digits $\tilde{U}_{\Pi_A(b)}^n$. If it finds at least one such sequence, it chooses one of these sequences and transmits the syndrome $Sx_b \triangleq H_{2b} \tilde{U}^n_{\pi_A(b)}$ to the decoder. If it finds no such sequence, it declares an encoding error. The operation of the $Y$-encoder is similar.

Let $\psi_b(\cdot,\cdot)$ be the decoder corresponding to the good channel code $\mc{C}_{2b}$. The decoder action is described by the following series of equations. The decoder receives the syndromes $Sx_b$ and $Sy_b$.
\begin{align}
\nonumber \hat{\tilde{Z}}_{\pi_A(b)} &= \psi_b \left(Sx_b \oplus_{{p_b}^{r_b}} Sy_b, \tilde{Z}^n_{\Pi_A(b)} \right) \\
\nonumber &= \psi_b \left(H_{2b} \tilde{U}^n_{\pi_A(b)} \oplus_{{p_b}^{r_b}} H_{2b} \tilde{V}^n_{\pi_A(b)}, \tilde{Z}^n_{\Pi_A(b)}\right)\\
\nonumber &= \psi_b \left(H_{2b} \left(\tilde{U}^n_{\pi_A(b)} \oplus_{{p_b}^{r_b}} \tilde{V}^n_{\pi_A(b)} \right), \tilde{Z}^n_{\Pi_A(b)} \right) \\
\nonumber &= \psi_b \left(H_{2b} \tilde{Z}^n_{\pi_A(b)}, \tilde{Z}^n_{\Pi_A(b)} \right) \\
&\stackrel{(a)}= \tilde{Z}^n_{\pi_A(b)} \quad \mbox{with high probability}
\end{align}
where (a) follows from the fact that $\mc{C}_{2b}$ is a good channel code for the triple $(\mc{\tilde{Z}}_{\pi_A(b)}, \tilde{\mc{Z}}_{\Pi_A(b)}, P_{\tilde{Z}_{\pi_A(b)} \tilde{Z}_{\Pi_A(b)}})$.

The rate expended by the $X$-encoder at the $b$th stage can be calculated as follows. Since $\mc{C}_{11b}$ is a good source code for the triple $(\mc{X} \times \tilde{\mc{U}}_{\Pi_A(b)}, \tilde{\mc{U}}_{\pi_A(b)}, P_{X \tilde{U}_{\Pi_A(b)} \tilde{U}_{\pi_A(b)}})$, we have from equation (\ref{eq:goodlinsrclimiteq}) that the dimensions of the parity check matrix $H_{11b}$ satisfy
\begin{equation}
\frac{k_{11b}}{n} \log p_b^{r_b} \leq \min \left( H(\tilde{U}_{\pi_A(b)} \mid X, \tilde{U}_{\Pi_A(b)}), r_b(|H(\tilde{U}_{\pi_A(b)} \mid X, \tilde{U}_{\Pi_A(b)}) - \log p_b^{r_b-1}|^{+}) \right) - \epsilon_1
\end{equation}
Since $\mc{C}_{2b}$ is a good channel code for the triple $(\tilde{\mc{Z}}_{\pi_A(b)}, \tilde{\mc{Z}}_{\Pi_A(b)}, P_{\tilde{Z}_{\pi_A(b)} \tilde{Z}_{\Pi_A(b)}})$, the dimensions of the parity check matrix $H_{2b}$ satisfy
\begin{equation}
\frac{k_{2b}}{n} \log p_b^{r_b} \geq \max_{0 \leq i < r_b} \left( \frac{r_b}{r_b-i} \right) \left(H(\tilde{Z}_{\pi_A(b)} \mid \tilde{Z}_{\Pi_A(b)}) - H([\tilde{Z}_{\pi_A(b)}]_i | \tilde{Z}_{\Pi_A(b)}) \right) + \epsilon_2
\end{equation}
The rate of the nested group code in bits would be $R_1 = n^{-1} (k_{2b} - k_{11b}) \log p_b^{r_b}$. Therefore,
\begin{align} \label{eq:R1b1}
\nonumber R_{1b}^{(1)} &\geq \left[ \max_{0 \leq i < r_b} \left( \frac{r_b}{r_b-i} \right) \left(H(\tilde{Z}_{\pi_A(b)} \mid \tilde{Z}_{\Pi_A(b)}) - H([\tilde{Z}_{\pi_A(b)}]_i | \tilde{Z}_{\Pi_A(b)}) \right) \right] \\ &- \left[ \min \left( H(\tilde{U}_{\pi_A(b)} \mid X, \tilde{U}_{\Pi_A(b)}), r_b(|H(\tilde{U}_{\pi_A(b)} \mid X, \tilde{U}_{\Pi_A(b)}) - \log p_b^{r_b-1}|^{+}) \right) \right] + \epsilon_1 + \epsilon_2
\end{align}

The other option that the encoders have is to directly encode the digits $\tilde{U}_{\pi_A(b)}$ and  $\tilde{V}_{\pi_A(b)}$. This can also be accomplished using nested group codes as follows. The $X$ encoder uses the nested group code $(\mc{C}_{11b}, \mc{C}_{21b})$ such that the fine group code $\mc{C}_{11b}$ is a good source code for the triple $(\mc{X} \times \tilde{\mc{U}}_{\Pi_A(b)}, \tilde{\mc{U}}_{\pi_A(b)}, P_{X \tilde{U}_{\Pi_A(b)} \tilde{U}_{\pi_A(b)}})$ and $\mc{C}_{21b}$ is a good channel code for the triple $(\tilde{\mc{U}}_{\pi_A(b)}, \tilde{\mc{Z}}_{\Pi_A(b)}, P_{\tilde{U}_{\pi_A(b)} \tilde{Z}_{\Pi_A(b)}})$. The Y encoder uses the nested group code $(\mc{C}_{12b}, \mc{C}_{22b})$ such that the fine group code $\mc{C}_{12b}$ is a good source code for the triple $(\mc{Y} \times \tilde{\mc{V}}_{\Pi_A(b)}, \tilde{\mc{V}}_{\pi_A(b)}, P_{Y \tilde{V}_{\Pi_A(b)} \tilde{V}_{\pi_A(b)}})$ and $\mc{C}_{22b}$ is a good channel code for the triple $(\tilde{\mc{V}}_{\pi_A(b)}, \tilde{\mc{Z}}_{\Pi_A(b)}, P_{\tilde{V}_{\pi_A(b)} \tilde{Z}_{\Pi_A(b)}})$. The encoding operation is similar to that described earlier and it is easy to verify its correctness.

The rate of this nested group code in bits would be $R_1 = n^{-1} (k_{2b} - k_{11b}) \log p_b^{r_b}$. Therefore,
\begin{align} \label{eq:R1b2}
\nonumber R_{1b}^{(2)} &\geq \left[ \max_{0 \leq i < r_b} \left( \frac{r_b}{r_b-i} \right) \left(H(\tilde{U}_{\pi_A(b)} \mid \tilde{Z}_{\Pi_A(b)}) - H([\tilde{U}_{\pi_A(b)}]_i \mid \tilde{Z}_{\Pi_A(b)}) \right) \right] \\ &- \left[ \min \left( H(\tilde{U}_{\pi_A(b)} \mid X, \tilde{U}_{\Pi_A(b)}), r_b(|H(\tilde{U}_{\pi_A(b)} \mid X, \tilde{U}_{\Pi_A(b)}) - \log p_b^{r_b-1}|^{+}) \right) \right] + \epsilon_1 + \epsilon_2
\end{align}
Combining equations (\ref{eq:R1b1}) and (\ref{eq:R1b2}), we have proved Theorem \ref{thm:mainthm}.
\end{proof}

\noindent \textbf{Remark 1}: The design of the channel code used in the above derivation assumes that the side information available to the decoder at the $b$th stage is $\tilde{Z}_{\Pi_A(b)}$. However, it is possible that at some stage $1 \leq i \leq k$, the encoding was done in such a way that the decoder could decode $(\tilde{U}_{\pi_A(i)}, \tilde{V}_{\pi_A(i)})$ and not just $\tilde{Z}_{\pi_A(b)}$. Taking such considerations into account while designing the channel code for the $b$th stage would lead to a possible improvement of the rate region in Theorem \ref{thm:mainthm}.

\noindent \textbf{Remark 2}: In the above derivation, if the encoders choose to encode the sources $\tilde{U}_{\pi_A(b)}, \tilde{V}_{\pi_A(b)}$ directly instead of encoding the function $\tilde{Z}_{\pi_A(b)}$, further rate gains are possible when one encoder encodes its source conditional on the other source in addition to the side information already available at the decoder. Such improvements are omitted for the sake of clarity of the expressions constituting the definition of the achievable rate region.

\noindent \textbf{Remark 3}: The above coding theorem can be extended to the case of multiple distortion constraints in a straightforward fashion.

\section{Special cases} \label{sec:specialcases}

In this section, we consider the various special cases of the rate region presented in Theorem \ref{thm:mainthm}.

\subsection{Lossless Source Coding using Group Codes} \label{subsec:losslesssrccor}
We start by demonstrating the achievable rates using codes over groups for the problem of lossless source coding. A good group channel code $\mc{C}$ for the triple $(\mc{X}, 0, P_X)$ as defined in Definition \ref{defi:goodchcodedefi} can be used to achieve lossless source coding of the source $X$. The source encoder outputs $Hx^n$ where $H$ is the $k \times n$ parity check matrix of $\mc{C}$. The decoder uses the associated decoding function $\psi(\cdot,\cdot)$ to recover $\psi(Hx^n,0) = x^n$ with high probability. From equation (\ref{eq:goodlinchannellimiteq}), it follows that the dimensions of the parity check matrix satisfy
\begin{equation}
\frac{k}{n} \log p^r \geq \max_{0 \leq i < r} \left( \frac{r}{r-i} \right) (H(X) - H([X]_i))
\end{equation}
Recognizing the term in the left as the rate of the coding scheme, we get the following corollary to Theorem \ref{thm:mainthm}.
\begin{cor} \label{cor:losslesscodingcor}
Suppose $X$ is a non redundant random variable over the group $\ringr{}$ and the decoder wants to reconstruct $X$ losslessly. Then, there exists a group based coding scheme that achieves the rate
\begin{equation} \label{eq:groupentbound}
R \geq \max_{0 \leq i < r} \left( \frac{r}{r-i} \right) (H(X) - H([X]_i))
\end{equation}
\end{cor}
Putting $r=1$ in equation (\ref{eq:groupentbound}) reduces it to the well known result that linear codes over prime fields can compress a source down to its entropy. Thus, this achievable rate region using group codes can be strictly greater than Shannon entropy. A sufficient condition for the existence of group codes that attain the entropy bound is that
\begin{equation}
H([X]_i) \geq \frac{i}{r} H(X) \quad \mbox{for } 0 < i < r
\end{equation}

\subsection{Lossy Source Coding using Group Codes} \label{subsec:lossysrccor}
We next consider the case of lossy point to point source coding using codes built over the group $\ringr{}$. Consider a memoryless source $X$ with distribution $P_X$. The decoder attempts to reconstruct $U$ that is within distortion $D$ of $X$ as specified by some additive distortion measure $d \colon \mc{X} \times \mc{U} \rightarrow \Bbb{R}^{+}$. Suppose $U$ takes its values from the group $\ringr{}$. A good group source code $\mc{C}$ for the triple $(\mc{X},\mc{U},P_{XU})$ as defined in Definition \ref{defi:goodsrccodedefi} can be used to achieve lossy coding of the source $X$ provided the joint distribution $P_{XU}$ is such that $\Bbb{E}(d(X,U)) \leq D$ and $U$ is non-redundant. The source encoder outputs $u^n \in \mc{C}$ that is jointly typical with the source sequence $x^n$. An encoding error is declared if no such $u^n$ is found. The decoder uses $u^n$ as its reconstruction of the source $x^n$. From equation (\ref{eq:goodlinsrclimiteq}), it follows that the dimensions  of the parity check matrix associated with $\mc{C}$ satisfy
\begin{equation}
\frac{k}{n} \log p^r \leq \min (H(U|X), r|H(U|X) - \log p^{r-1}|^{+}).
\end{equation}
The rate of this encoding scheme is $R = \left( 1 - \frac{k}{n} \right) \log p^r$. Thus, we get the following corollary to Theorem \ref{thm:mainthm}.
\begin{cor} \label{cor:lossycodingcor}
Let $X$ be a discrete memoryless source and $\mc{U}$ be the reconstruction alphabet. Suppose $\mc{U} = \ringr{}$ and the decoder wants to reconstruct the source to within distortion $D$ as measured by the fidelity criterion $d(\cdot,\cdot)$. Without loss of generality, assume that $U$ is non-redundant. Then, there exists a group based coding scheme that achieves the rate
\begin{equation} \label{eq:groupRDbound1}
R \geq \min_{\stackrel{P_{U|X}}{\Bbb{E}d(X,U) \leq D}} \log p^r - \left( \min \left(H(U|X), r |H(U|X) - \log p^{r-1}|^{+} \right) \right).
\end{equation}
\end{cor}
If $U$ takes values in a general abelian group of order $n$ that is not necessarily a primary cyclic group, then a decomposition based approach similar to the one used in the proof of Theorem \ref{thm:mainthm} can be used. Suppose $n = \prod_{i=1}^{k} p_i^{e_i}$ is the prime factorization of $n$. Then, the group in which $U$ takes values is isomorphic to $\oplus_{i=1}^k \ring{p_i^{e_i}}$ where $p_i$ are not necessarily distinct primes. The random variable $U$ can be decomposed into its constituent digits $(U_1,\dots,U_k)$ which can then be encoded sequentially. The achievable rate can be obtained in a straightforward way. A simplification occurs when $\mc{U}$ is treated as a subset of the group $\oplus_{i=1}^k \ring{p_i}$ where $p_i$ are prime. In this case, the following rate-distortion bound is obtained.
\begin{cor} \label{cor:linsrccodecor}
Let $X$ be a discrete memoryless source and $\mc{U}$ be the reconstruction alphabet. Let $p_1,\dots,p_k$ be primes such that $\prod_{i=1}^k p_i \geq |\mc{U}|$. Suppose the decoder wants to reconstruct the source to within distortion $D$ as measured by the fidelity criterion $d(\cdot,\cdot)$. Then, there exists a group based coding scheme that achieves the rate
\begin{equation} \label{eq:groupRDbound2}
R \geq \min_{\stackrel{P_{U|X}}{\Bbb{E}d(X,U) \leq D}} \left( \sum_{i=1}^k \log p_i \right) - H(U|X)
\end{equation}
\end{cor}

Corollary \ref{cor:linsrccodecor} can be viewed as providing an achievable rate-distortion pair for lossy source coding using linear codes built over Galois fields. Note that it is possible to construct codebooks with rate $R = H(U) - H(U|X)$ by choosing codewords independently and uniformly from the set $\typset{}{U}$. By imposing the group structure on the codebook, we incur a rate loss of $( \sum_{i=1}^k \log p_i - H(U))$ bits per sample. This rate loss is strictly positive unless the random variable $U$ is uniformly distributed over $\oplus_{i=1}^k \ring{p_i}$.

\subsection{Nested Linear Codes}
We specialize the rate region of Theorem \ref{thm:mainthm} to the case when the nested group codes are built over cyclic groups of prime order, i.e., over Galois fields of prime order. In this case, group codes over $\ringr{}$ reduce to the well known linear codes over prime fields. It was already shown in Sections \ref{subsec:losslesssrccor} and \ref{subsec:lossysrccor} that Lemmas \ref{lemma:goodchcodelemma} and \ref{lemma:goodsrccodelemma} imply that linear codes achieve the entropy bound and incur a rate loss while used in lossy source coding. In this section, we demonstrate the implications of Theorem \ref{thm:mainthm} when specialized to the case of nested linear codes, i.e., when $r$ is set to $1$.

\subsubsection{Shannon Rate-Distortion Function}
We remark that Theorem \ref{thm:mainthm} shows the existence of nested linear codes that can be used to approach the rate-distortion bound in the single-user setting for arbitrary discrete sources and arbitrary distortion measures.

\begin{cor} \label{cor:RDcor}
Let $X$ be a discrete memoryless source with distribution $P_X$ and let $\mc{\hat{X}}$ be the reconstruction alphabet. Let the fidelity criterion be given by $d \colon \mc{X} \times \mc{\hat{X}} \rightarrow \Bbb{R}^{+}$. Then, there exists a nested linear code $(\mc{C}_1,\mc{C}_2)$ that achieves the rate-distortion bound
\begin{equation}
R(D) = \min_{\stackrel{P_{\hat{X}|X}}{\Bbb{E}d(X,\hat{X}) \leq D}} I(X;\hat{X})
\end{equation}
\end{cor}
\begin{proof}[\textbf{Proof}:]
Let the optimal forward test channel that achieves the bound be given by $P_{\hat{X}|X}$. Suppose $q$ is a prime such that $\mc{\hat{X}} \subset \ring{q}$ and $\hat{X}$ is non-redundant. The rate bound, given by $I(X;\hat{X})$ can be approached using a nested linear code $(\mc{C}_1,\mc{C}_2)$ built over the group $\ring{q}$. Here $\mc{C}_1$ is a good source code for the triple $(\mc{X},\mc{\hat{X}}, P_{X, \hat{X}})$ and $\mc{C}_2$ is a good channel code for the triple $(\mc{\hat{X}},\mc{S}, P_{\hat{X}S})$ where $\mc{S} = \{0\}$ and $S$ is a degenerate random variable with $P_S(0) = 1$. It follows from Lemmas \ref{lemma:goodsrccodelemma} and \ref{lemma:goodchcodelemma} that the dimensions of the parity check matrices associated with $\mc{C}_1$ and $\mc{C}_2$ satisfy
\begin{align}
\lim_{n \rightarrow \infty} \frac{k_1(n)}{n} \log q &= H(\hat{X} | X) \\
\lim_{n \rightarrow \infty} \frac{k_2(n)}{n} \log q &= H(\hat{X})
\end{align}
Thus, the rate achieved by this scheme is given by $n^{-1} (k_2(n)-k_1(n)) \log q = I(X;\hat{X})$.
\end{proof}
This can be intuitively interpreted as follows. For a code to approach the optimal rate-distortion function, the ``Voronoi'' region (under an appropriate encoding rule)  of most of the codewords should have a certain shape (say, shape A), and a high-probability set of codewords  should  be bounded in a region that has a certain shape (say, shape B). We choose $\mc{C}_1$ such that the ``Voronoi'' region (under the joint typicality encoding operation with respect to $p_{\hat{X},X}$) of each codeword has shape A. $\mc{C}_2$ is chosen such that its ``Voronoi'' region has shape B. Hence the set of ``coset leaders'' of $\mc{C}_1$ in $\mc{C}_2$ forms a code that can approach the optimal rate-distortion function. This reminds us of a similar phenomenon first observed in the case of Gaussian sources with mean squared error criterion in \cite{eyuboglu93}, where the performance of a quantizer is measured by so-called granular gain and boundary gain. Granular gain measures how closely the Voronoi regions of the codewords approach a sphere, and boundary gain measures how closely  the boundary region approaches a sphere.

\subsubsection{Berger-Tung Rate Region}
We now show that Theorem \ref{thm:mainthm} implies that nested linear codes built over prime fields can achieve the rate region of the Berger-Tung based coding scheme presented in Lemma \ref{lemma:BTschemelemma}.

\begin{cor} \label{cor:BTcor}
Suppose we have a pair of correlated discrete sources $(X,Y)$ and the decoder is interested in reconstructing $\hat{Z}$ to within distortion $D$ as measured by a fidelity criterion $d \colon \mc{X} \times \mc{Y} \times \mc{\hat{Z}} \rightarrow \Bbb{R}^{+}$. For this problem, an achievable rate region using nested linear codes is given by
\begin{align} \label{eq:BTrateregion}
\nonumber \mc{RD}_{BT} &= \bigcup_{(P_{U|X},P_{V|Y}) \in   \mc{P}} \left\{ (R_1,R_2) \colon R_1 \geq I(X;U|Y), \right. \\
&\left. R_2 \geq I(Y;V|X), R_1 + R_2 \geq I(X;U)+I(Y;V)-I(U;V) \right\}
\end{align}
where $\mc{P}$ is the family of all joint distributions $P_{XYUV}$ that satisfy the Markov chain $U-X-Y-V$ such that the distortion criterion $\Bbb{E}d(X,Y,\hat{Z}(U,V)) \leq D$ is met. Here $\hat{Z}(U,V)$ is the optimal reconstruction of $\hat{Z}$ with respect to the distortion criterion given $U$ and $V$.
\end{cor}

\begin{proof}[\textbf{Proof}:]
We proceed by first reconstructing the function $G(U,V) = (U,V)$ at the decoder and then computing the function $\hat{Z}(U,V)$. For ease of exposition, assume that $\mc{U} = \mc{V} = \ring{q}$ for some prime $q$. If they are not, a decomposition based approach can be used and the proof is similar to the one presented below. Clearly, $G(U,V)$ can be embedded in the abelian group $A \triangleq \ring{q} \oplus \ring{q}$. The associated mappings are given by $\tilde{U} = (U,0)$ and $\tilde{V} = (0,V)$ where $0$ is the identity element in $\ring{q}$. Thus, $\tilde{Z}_1 = U + 0 = U$ and $\tilde{Z}_2 = 0+V = V$. Encoding is done in two stages. Let the permutation $\pi_A(\cdot)$ be the identity permutation. Substituting this into equations (\ref{eq:R1b1}) and (\ref{eq:R1b2}) gives us
\begin{align}
\nonumber R_{11} &\geq \min \{ H(\tilde{Z}_1), H(\tilde{U}_1) \} - H(\tilde{U}_1 | X) = I(X;\tilde{U}_1) = I(X;U), \\
\nonumber R_{21} &\geq \min \{ H(\tilde{Z}_1), H(\tilde{V}_1) \} - H(\tilde{V}_1 | Y) = 0, \\
\nonumber R_{12} &\geq \min \{ H(\tilde{Z}_2 \mid \tilde{Z}_1), H(\tilde{U}_2 \mid \tilde{U}_1) \} - H( \tilde{U}_2 \mid X, \tilde{U}_1) = 0, \\
\nonumber R_{22} &\geq \min \{ H(\tilde{Z}_2 \mid \tilde{Z}_1), H(\tilde{V}_2 \mid \tilde{V}_1) \} - H( \tilde{V}_2 \mid Y, \tilde{V}_1) \\
\nonumber &= H(\tilde{Z}_2 \mid \tilde{Z}_1) - H(\tilde{V}_2 \mid Y, \tilde{V}_1) = H(V|U) - H(V|Y) \\
&= I(Y;V|U)
\end{align}
This is one of the corner points of the rate region given in equation (\ref{eq:BTrateregion}). Choosing the permutation $\pi_A(\cdot)$ to be the derangement gives us the other corner point and time sharing between the two points yields the entire rate region of equation (\ref{eq:BTrateregion}). The rate needed to reconstruct $U,V$ at the decoder coincides with the Berger-Tung rate region \cite{berger-tung,berger77}.
\end{proof}
We note that this implies that our theorem recovers the rate regions of the problems considered by Wyner and Ziv \cite{wyner-ziv}, Ahlswede-Korner-Wyner \cite{ahlswede-korner,wyner75}, Berger and Yeung \cite{yeung-berger} and Slepian and Wolf \cite{slepian-wolf} since the Berger-Tung problem encompasses all these problems as special cases.

\subsection{Lossless Reconstruction of Modulo-$2$ Sum of Binary Sources}
In this section, we show that Theorem \ref{thm:mainthm} recovers the rate region derived by Korner and Marton \cite{korner-marton} for the reconstruction of the modulo-$2$ sum of two binary sources. Let $X,Y$ be correlated binary sources. Let the decoder be interested in reconstructing the function $F(X,Y) = X \oplus_2 Y$ losslessly. In this case, the auxiliary random variables can be chosen as $U = X, V = Y$. Clearly, this function can be embedded in the groups $\ring{2}, \ring{3}, \ring{4}$ and $\ring{2} \oplus \ring{2}$. For embedding in $\ring{2}$, the rate region of Theorem \ref{thm:mainthm} reduces to
\begin{equation}
R_1 \geq \min(H(X),H(X \oplus_2 Y)), \quad R_2 \geq \min(H(Y),H(X \oplus_2 Y))
\end{equation}
It can be verified that embedding in $\ring{3}$ or $\ring{4}$ always gives a worse rate than embedding in $\ring{2}$. Embedding in $\ring{2} \oplus \ring{2}$ results in the Slepian-Wolf rate region. Combining these rate regions, we see that a sum rate of $R_1 + R_2 = \min ( 2H(X \oplus_2 Y), H(X,Y))$ is achievable using our coding scheme. This recovers the Korner-Marton rate region for this problem \cite{korner-marton,csiszarbook}.  Moreover, one can also show that this approach can recover the Ahlswede-Han rate region \cite{ahlswede-han} for this problem, which is an improvement over the Korner-Marton region.

\section{Examples} \label{sec:examples}

In this section, we consider examples of the coding theorem (Theorem \ref{thm:mainthm}). First we consider the problem of losslessly reconstructing a function of correlated quaternary sources. We then derive an achievable rate region for the case when the decoder is interested in the modulo-$2$ sum of two binary sources to within a Hamming distortion of $D$.

\subsection{Lossless Encoding of a Quaternary Function}

Consider the following distributed source coding problem. Let $(X,Y)$ be correlated random variables both taking values in $\ring{4}$. Let $X,Z$ be independent random variables taking values in $\ring{4}$ according to the distributions $P_X$ and $P_Z$ respectively. Define $p_i \triangleq P_X(i), q_i \triangleq P_Z(i)$ for $i = 0,\dots,3$. Assume further that the random variable $Z$ is non-redundant, i.e., $q_1 + q_3 > 0$. Define the random variable $Y$ as $Y = X \oplus_4 Z$. Suppose $X$ and $Y$ are observed by two separate encoders which communicate their quantized observations to a central decoder. The decoder is interested in reconstructing the function $Z = (X-Y) \mod 4$ losslessly.

Since we are interested in lossless reconstruction, we can choose the auxiliary random variables $U,V$ to be $U = X, V = Y$. The function $G(U,V)$ then reduces to $F(X,Y) \triangleq (X-Y) \mod 4$. This function can be embedded in several groups with order less than or equal to $16$.  We claim that this function $F(X,Y)$ can be embedded in the groups $\ring{4}, \ring{7}, \ring{2} \oplus \ring{2} \oplus \ring{2}$ and $\ring{4} \oplus \ring{4}$ with nontrivial performance.  For each of these groups, we compute the achievable rate as given by Theorem \ref{thm:mainthm} below. For simplicity, we restrict ourselves to the rate regions given by equation (\ref{eq:R1b1}) alone.

Lets consider the group $\ring{4}$ first. Define the mappings $\tilde{x} \triangleq S_X^{(\ring{4})}(x) = x \mbox{ for all } x \in \ring{4}, \tilde{y} \triangleq S_Y^{(\ring{4})}(y) = -y \mbox{ for all } y \in \ring{4}$ and $S_F^{(\ring{4})}(z) = z \mbox{ for all } z \in \ring{4}$. With these mappings, it follows from Definition \ref{defi:embeddingdefi} that $F(X,Y)$ is embeddable in $\ring{4}$ with respect to the distribution $P_{XY}$. From Theorem \ref{thm:mainthm}, it follows that an achievable rate region using this embedding is given by
\begin{align}
\nonumber R_1 &\geq \max \{ H(Z), 2 (H(Z) - H([Z]_1)) \} \\
&= \max \{ h(q_0,q_1,q_2,q_3), 2 (h(q_0,q_1,q_2,q_3) - h(q_0 + q_2, q_1 + q_3)) \} \label{eq:R1Z4} \\
\nonumber R_2  &\geq \max \{ H(Z), 2 (H(Z) - H([Z]_1)) \} \\
&= \max \{ h(q_0,q_1,q_2,q_3), 2 (h(q_0,q_1,q_2,q_3) - h(q_0 + q_2, q_1 + q_3)) \} \label{eq:R2Z4}
\end{align}
giving a sum rate of
\begin{equation} \label{eq:Rsumz4}
R_{\ring{4}} \triangleq R_1 + R_2 \geq 2 \max \{ h(q_0,q_1,q_2,q_3), 2 (h(q_0,q_1,q_2,q_3) - h(q_0 + q_2, q_1 + q_3)) \}
\end{equation}

It can be verified that $F(X,Y)$ can't be embedded in $\ring{5}$ or $\ring{6}$. It can be embedded in $\ring{7}$ with the following mappings. Define $\tilde{x} \triangleq S_X^{(\ring{7})}(x) = x$ for all $x \in \ring{4}$, $\tilde{y} \triangleq S_Y^{(\ring{7})}(y) = -y$ for all $y \in \ring{4}$ where $-y$ is the additive inverse of $y$ in $\ring{7}$ and $S_F^{(\ring{7})}(0) = 0, S_F^{(\ring{7})}(1) = S_F^{(\ring{7})}(4) = 1, S_F^{(\ring{7})}(2) = S_F^{(\ring{7})}(5) = 2, S_F^{(\ring{7})}(3) = S_F^{(\ring{7})}(6) = 3$. Let $Z = \tilde{X} \oplus_7 \tilde{Y}$. From Theorem \ref{thm:mainthm}, it follows that an achievable rate region using this embedding is given by
\begin{align}
R_1 &\geq H(Z) = h(q_0, (1-p_0)q_3, (1-p_0-p_1)q_2, p_3 q_1, p_0 q_3, (p_0+p_1)q_2, (1-p_3)q_1) \label{eq:R1Z7} \\
R_2 &\geq H(Z) = h(q_0, (1-p_0)q_3, (1-p_0-p_1)q_2, p_3 q_1, p_0 q_3, (p_0+p_1)q_2, (1-p_3)q_1) \label{eq:R2Z7}
\end{align}
giving a sum rate of
\begin{equation} \label{eq:Rsumz7}
R_{\ring{7}} \triangleq R_1 + R_2 \geq 2 h(q_0, (1-p_0)q_3, (1-p_0-p_1)q_2, p_3 q_1, p_0 q_3, (p_0+p_1)q_2, (1-p_3)q_1)
\end{equation}

Of the three abelian groups of order $8$, it can be verified that embedding $F(X,Y)$ in $\ring{8}$ results in the same rate region as given by equations (\ref{eq:R1Z7}) and (\ref{eq:R2Z7}) and embedding $F(X,Y)$ in $\ring{2} \oplus \ring{4}$ results in the same rate region as given by equations (\ref{eq:R1Z4}) and (\ref{eq:R2Z4}). So, we consider embedding $F(X,Y)$ in $\ring{2} \oplus \ring{2} \oplus \ring{2}$. Recall that elements of the abelian group $\ring{2} \oplus \ring{2} \oplus \ring{2}$ can be treated as $3$ bit vectors over $\ring{2}$. The mappings $S_X^{(\ring{2} \oplus \ring{2} \oplus
  \ring{2})}(\cdot), S_Y^{(\ring{2} \oplus \ring{2} \oplus
  \ring{2})}(\cdot)$ and $S_F^{(\ring{2} \oplus \ring{2} \oplus
  \ring{2})}(\cdot)$ are as given in Table
\ref{table:embedz2z2z2table}.

\begin{table}
\centering \subtable[$S_X(\cdot)$]{
    \begin{tabular} {|c||c|}
    \hline
    $X$ & $\tilde{X}$ \\
    \hline & \\[-1.35em] \hline
    0 & 000\\
    \hline
    1 & 001\\
    \hline
    2 & 100\\
    \hline
    3 & 101\\
    \hline
    \end{tabular}
    \label{table:xmapping}
} \qquad \qquad \subtable[$S_Y(\cdot)$]{
    \begin{tabular} {|c||c|}
    \hline
    $Y$ & $\tilde{Y}$ \\
    \hline & \\[-1.35em] \hline
    0 & 000\\
    \hline
    1 & 010\\
    \hline
    2 & 100\\
    \hline
    3 & 110\\
    \hline
    \end{tabular}
    \label{table:ymapping}
} \qquad \qquad \subtable[$S_F(\cdot)$]{
    \begin{tabular} {|c||c|}
    \hline
    $z$ & $S_F(z)$ \\
    \hline & \\[-1.35em] \hline
    000,011 & 0\\
    \hline
    001,110 & 1\\
    \hline
    100,111 & 2\\
    \hline
    010,101 & 3\\
    \hline
    \end{tabular}
    \label{table:funcmapping}
} \caption{Mappings for embedding $F(X,Y)$ in $\ring{2} \oplus \ring{2} \oplus \ring{2}$} \label{table:embedz2z2z2table}
\end{table}

Define the random variable $\tilde{Z} = \tilde{U} \oplus \tilde{V}$ where $\oplus$ is addition in $\ring{2} \oplus \ring{2} \oplus \ring{2}$. With these mappings, an achievable rate region can be derived using Theorem \ref{thm:mainthm} as below. Choose the permutation $\pi_{\ring{2} \oplus \ring{2} \oplus \ring{2}} (\cdot) $ as $\pi(1) = 2, \pi(2) = 3, \pi(3) = 1$. Encoding is carried out in $3$ stages with the corresponding rates being
\begin{align}
R_{11} &= 0 , R_{21} = H(\tilde{Z}_2) \\
R_{12} &= H(\tilde{Z}_3 \mid \tilde{Z}_2) , R_{22} = 0 \\
R_{13} &= H(\tilde{Z}_1 \mid \tilde{Z}_2, \tilde{Z}_3) , R_{23} = H(\tilde{Z}_1 \mid \tilde{Z}_2, \tilde{Z}_3).
\end{align}

Summing over the $3$ stages of encoding, we get an achievable sum rate of $R_1 + R_2 \geq H(Z) + H(\tilde{Z}_1 \mid \tilde{Z}_2, \tilde{Z}_3) = 2H(Z) - H(\tilde{Z}_2, \tilde{Z}_3)$. In terms of $p_i, q_i$, this sum rate can be expressed as
\begin{equation} \label{eq:Rsumz2z2z2}
R_{\ring{2} \oplus \ring{2} \oplus \ring{2}} \triangleq R_1 + R_2 \geq 2h( p_{02}q_0, p_{13}q_3, p_{02}q_1, p_{13}q_0, p_{02}q_2, p_{13}q_1, p_{02}q_3, p_{13}q_2) - h(p_{02}q_{02}, p_{13}q_{13}, p_{02}q_{13}, p_{13}q_{02})
\end{equation}
where $p_{02} \triangleq p_0 + p_2, p_{13} \triangleq p_1 + p_3, q_{02} \triangleq q_0 + q_2$ and $q_{13} \triangleq q_1 + q_3$.

Embedding $F(X,Y)$ in groups of order $9$ to $15$ result in rate regions which are worse than the ones already derived. We next present an achievable rate region when $F(X,Y)$ is embedded in $\ring{4} \oplus \ring{4}$. We use the mappings $S_X^{(\ring{4} \oplus
  \ring{4})}(x) = x0$ for all $x \in \ring{4}$, $S_Y^{(\ring{4} \oplus
  \ring{4})}(y) = 0y$ for all $y \in \ring{4}$ and $S_F^{(\ring{4}
  \oplus \ring{4})}(xy) = (x,y)$ for all $(x,y) \in \ring{4}^2$. This
embedding corresponds to reconstructing the sources $X$ and $Y$ losslessly and the rate region coincides with the Slepian-Wolf rate region.
\begin{equation} \label{eq:Rsumz4z4}
R_{\ring{4} \oplus \ring{4}} \triangleq R_1 + R_2 \geq H(X,Y) = H(X) + H(Z) = h(p_0,p_1,p_2,p_3) + h(q_0,q_1,q_2,q_3)
\end{equation}

Combining equations (\ref{eq:Rsumz4}), (\ref{eq:Rsumz7}), (\ref{eq:Rsumz2z2z2}) and (\ref{eq:Rsumz4z4}) gives us an achievable rate region for this problem. Each of these achievable rate regions outperform the others for certain values of $P_X$ and $P_Z$. This is illustrated in Table \ref{table:bestdistribs}.


\begin{table} \begin{center}
\begin{tabular}{|c|c|c|c|c|c|} \hline
$P_X$ & $P_Z$ & $R_{\ring{4}}$ & $R_{\ring{7}}$ & $R_{\ring{2} \oplus \ring{2} \oplus \ring{2}}$ & $R_{\ring{4} \oplus \ring{4}}$ \\
\hline \hline $[\frac{1}{4} \, \frac{1}{4} \, \frac{1}{4} \, \frac{1}{4}]$ & $[\frac{1}{2} \, 0 \, \frac{1}{4} \, \frac{1}{4}]$ & $3$ & $3.9056$ & $3.1887$ & $3.5$
\\
\hline $[\frac{3}{10} \, \frac{6}{10} \, \frac{1}{10} \, 0]$ & $[0 \, \frac{4}{5} \, \frac{1}{20} \, \frac{3}{20}]$ & $2.3911$ & $2.0797$ & $2.4529$ &
$2.1796$ \\
\hline $[\frac{1}{3} \, \frac{1}{10} \, \frac{1}{2} \, \frac{1}{15}]$ & $[\frac{3}{7} \, \frac{1}{7} \, \frac{1}{7} \, \frac{2}{7}]$ & $3.6847$ & $4.5925$ &
$3.3495$ & $3.4633$ \\
\hline $[\frac{9}{10} \, \frac{1}{30} \, \frac{1}{30} \, \frac{1}{30}]$ & $[\frac{3}{20} \, \frac{3}{4} \, \frac{1}{20} \, \frac{1}{20}]$ & $2.308$ &
$2.7065$ & $1.9395$ & $1.7815$ \\
\hline
\end{tabular} \end{center}
\caption{Example distributions for which embedding in a given group gives the lowest sum rate.} \label{table:bestdistribs}
\end{table}

\subsection{Lossy Reconstruction of the Modulo-$2$ Sum of Binary Sources} \label{subsec:lossyxorex}

This example concerns the reconstruction of the binary XOR function with the Hamming distortion criterion. The rate region of Theorem \ref{thm:mainthm} is very cumbersome to calculate analytically in the general case. So, we restrict our attention to the case of symmetric source distribution and additive test channels in the derivation below where the intention is to demonstrate the analytical evaluation of the rate region of Theorem \ref{thm:mainthm}. We then present plots where the entire sum rate-distortion region is computed without any restrictive assumptions.

Consider a binary correlated source $(X,Y)$ with symmetric joint distribution $P_{XY}(0,0) = P_{XY}(1,1) = q/2$ and $P_{XY}(1,0) = P_{XY}(0,1) = p/2$. Suppose we are interested in reconstructing $F(X,Y) = X \oplus_2 Y$ within Hamming distortion $D$. We present an achievable rate pair for this problem based on Theorem \ref{thm:mainthm} and compare it to the achievable rate region presented in Lemma \ref{lemma:BTschemelemma}. It was shown in \cite{gu-jana-effros} that it suffices to restrict the cardinalities of the auxiliary random variables $U$ and $V$ to the cardinalities of their respective source alphabets in order to compute the Berger-Tung rate region. Since the scheme presented in Lemma \ref{lemma:BTschemelemma} is based on the Berger-Tung coding scheme, the rate region $\mc{RD}_{BT}$ for this problem can be computed by using binary auxiliary random variables.

Let us now evaluate the rate region provided by Theorem \ref{thm:mainthm} for this problem. The auxiliary random variables $U$ and $V$ are binary and suppose the test channel $P_{XY} P_{U|X} P_{V|Y}$ is fixed. The function $G(U,V)$ which is the optimal reconstruction of $X \oplus_2 Y$ given $U$ and $V$ can then be computed. In general, this function can take any of the $16$ possible values depending upon the test channel $P_{XY} P_{U|X} P_{V|Y}$.

Let us choose the auxiliary random variables $U$ and $V$ to be binary and for ease of exposition, let them be defined as $U = X \oplus_2 Q_1$ and $V = Y \oplus_2 Q_2$. Here $Q_1,Q_2$ are independent binary random variables with $P(Q_i = 0) = q_i, i = 1,2$. Let $p_i = 1 - q_i, i = 1,2$. Define $\alpha = q_1 q_2 + p_1 p_2$ and $\beta = 1-\alpha$. Once the test channel $P_{XY} P_{U|X} P_{V|Y}$ is thus fixed, the optimal reconstruction function $G(U,V)$ that minimizes the probability $P(F(X,Y) \neq G(U,V))$ can be computed. It can be showed that
\begin{align}
G(U,V) &= \left\{ \begin{array}{cc} 0 & \alpha > p, \alpha < q \\
U \oplus_2 V & \alpha > p, \alpha > q \\
\overline{U \oplus_2 V} & \alpha < p, \alpha < q \\
1 & \alpha < p, \alpha > q \end{array} \right.
\end{align}
where $\overline{a}$ denotes the complement of the bit $a$. The corresponding distortion for these reconstructions can be calculated as
\begin{align} \label{eq:lossyXORDeq}
D(\alpha) &= \left\{ \begin{array}{cc} p & \alpha > p, \alpha < q \\
\beta & \alpha > p, \alpha > q \\
\alpha & \alpha < p, \alpha < q \\
q & \alpha < p, \alpha > q \end{array} \right.
\end{align}

Clearly, no rate need be expended if the function to be reconstructed is $G(U,V) = 0$ or $G(U,V) = 1$. It is also easy to see that the rates needed would be the same for both $G(U,V) = U \oplus_2 V$ and $G(U,V) = \overline{U \oplus_2 V}$. Let us therefore consider only reconstructing $G(U,V) = U \oplus_2 V$. It can be shown that this function is embeddable in the groups $\ring{2}, \ring{3}, \ring{4}$ and $\ring{2} \oplus \ring{2}$. Let us consider the group $A \triangleq \ring{2}$. The associated mappings $S_U^{(A)}(\cdot), S_V^{(A)}(\cdot)$ and $S_G^{(A)}(\cdot)$ are all identity mappings. In this case, we have only one digit to encode. Further, note that $P(Z_1 = 0) = P(U_1 \oplus_2 V_1 = 0) = q\alpha + p \beta$.

The rates of the encoders are given by equations (\ref{eq:R1b1}) and (\ref{eq:R1b2}) to be
\begin{align}
\nonumber R_{11} &= \min \{ H(U_1), H(Z_1) \} - H(U_1 \mid X)
\\ \nonumber &= \min \{1, h(q\alpha+p\beta)\} -h(q_1)
\\ &= h(q\alpha+p\beta)-h(q_1) \\
\nonumber R_{21} &= \min \{ H(V_1), H(Z_1) \} - H(V_1 \mid Y)
\\ \nonumber &= \min \{1, h(q\alpha+p\beta)\}-h(q_2)
\\ &= h(q\alpha+p\beta)-h(q_2)
\end{align}
where $h(\cdot)$ is the binary entropy function. Thus, an achievable rate region for this problem is
\begin{align}
\mc{R} &= \bigcup_{0 \leq q_1,q_2 \leq 1} \left\{ (R_1,R_2,D) \colon R_1 \geq h(q\alpha + p \beta) - h(q_1), R_2 \geq h(q\alpha + p\beta) - h(q_2), D \geq D(\alpha) \right\}
\end{align}
where $D(\alpha)$ is given in equation (\ref{eq:lossyXORDeq}). Rate points achieved by embedding the function in the abelian groups $\ring{3}, \ring{4}$ are strictly worse than that achieved by embedding the function in $\ring{2}$ while embedding in $\ring{2} \oplus \ring{2}$ gives the Slepian-Wolf rate region for the lossless reconstruction of $(X,Y)$.

We now plot the entire sum rate-distortion region for the case of a general source distribution and general test channels $P_{U|X}, P_{V|Y}$ and compare it with the Berger-Tung rate region $\mc{RD}_{BT}$ of Fact \ref{lemma:BTschemelemma}.

\begin{figure}[htp]
\centering
\epsfig{figure=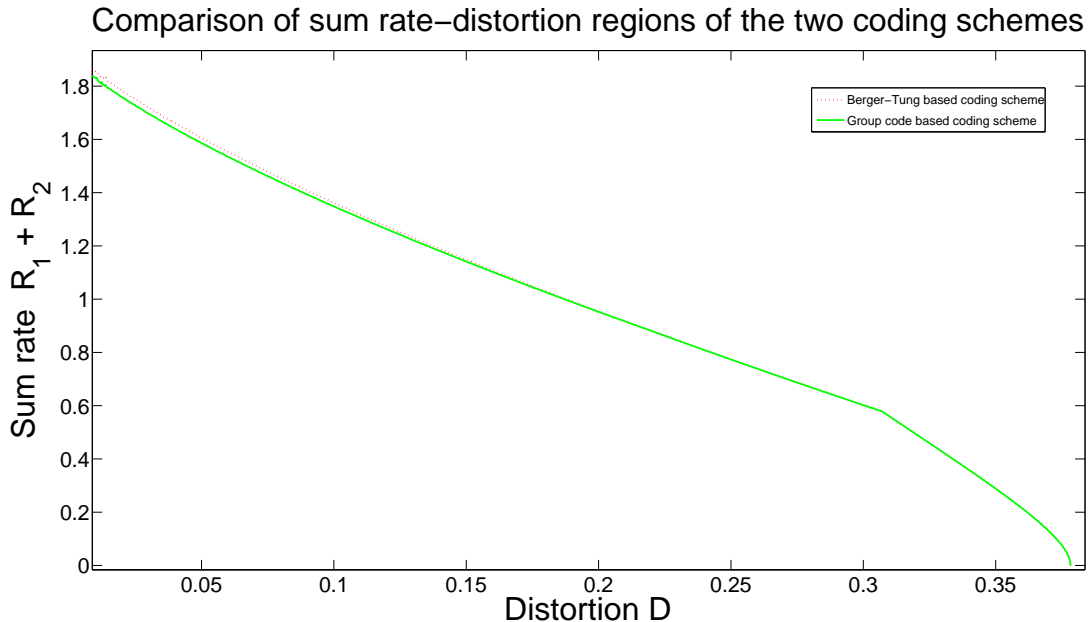, width = \textwidth}
\caption{Sum rate-distortion region for the distribution given in Table \ref{table:BT_example1}}
\label {fig:sumrate_noconv}
\end{figure}

\begin{figure}[htp]
\centering
\epsfig{figure = 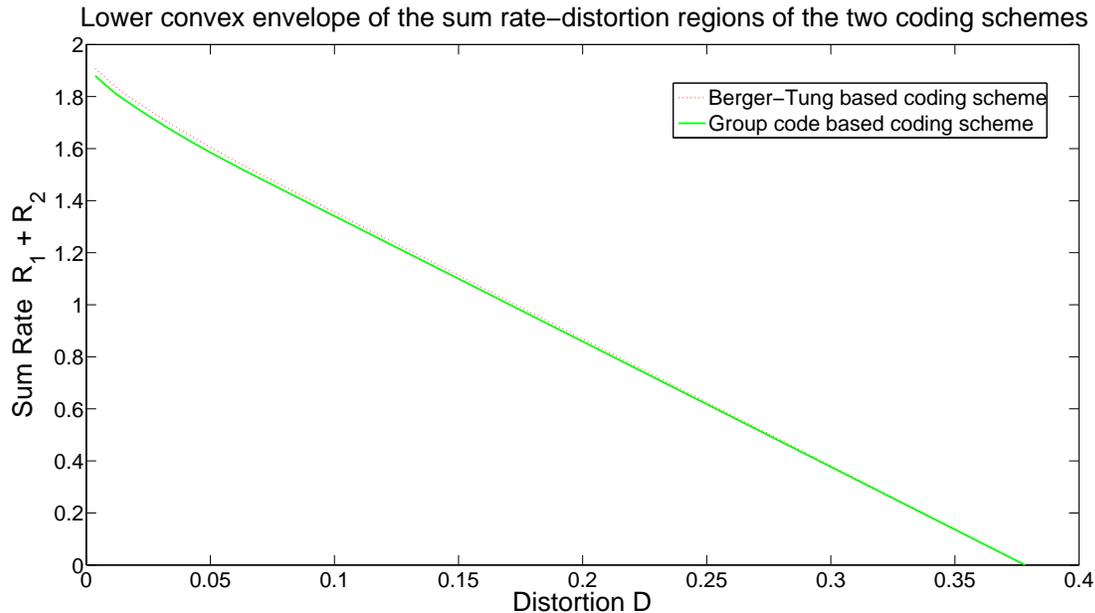, width = \textwidth}
\caption{Lower Convex envelope of the sum rate-distortion region}
\label{fig:sumrate_convhull}
\end{figure}

Figures \ref{fig:sumrate_noconv} and \ref{fig:sumrate_convhull} demonstrate that the sum rate-distortion regions of Theorem \ref{thm:mainthm} and Fact \ref{lemma:BTschemelemma}. Theorem \ref{thm:mainthm} offers improvements over the rate region of Fact \ref{lemma:BTschemelemma} for low distortions as shown more clearly in Figure \ref{fig:zoomedfigs}. The joint distribution of the sources used in this example is given in Table \ref{table:BT_example1}.

\begin{figure}[htp]
\centering
\subfigure{
\epsfig{figure=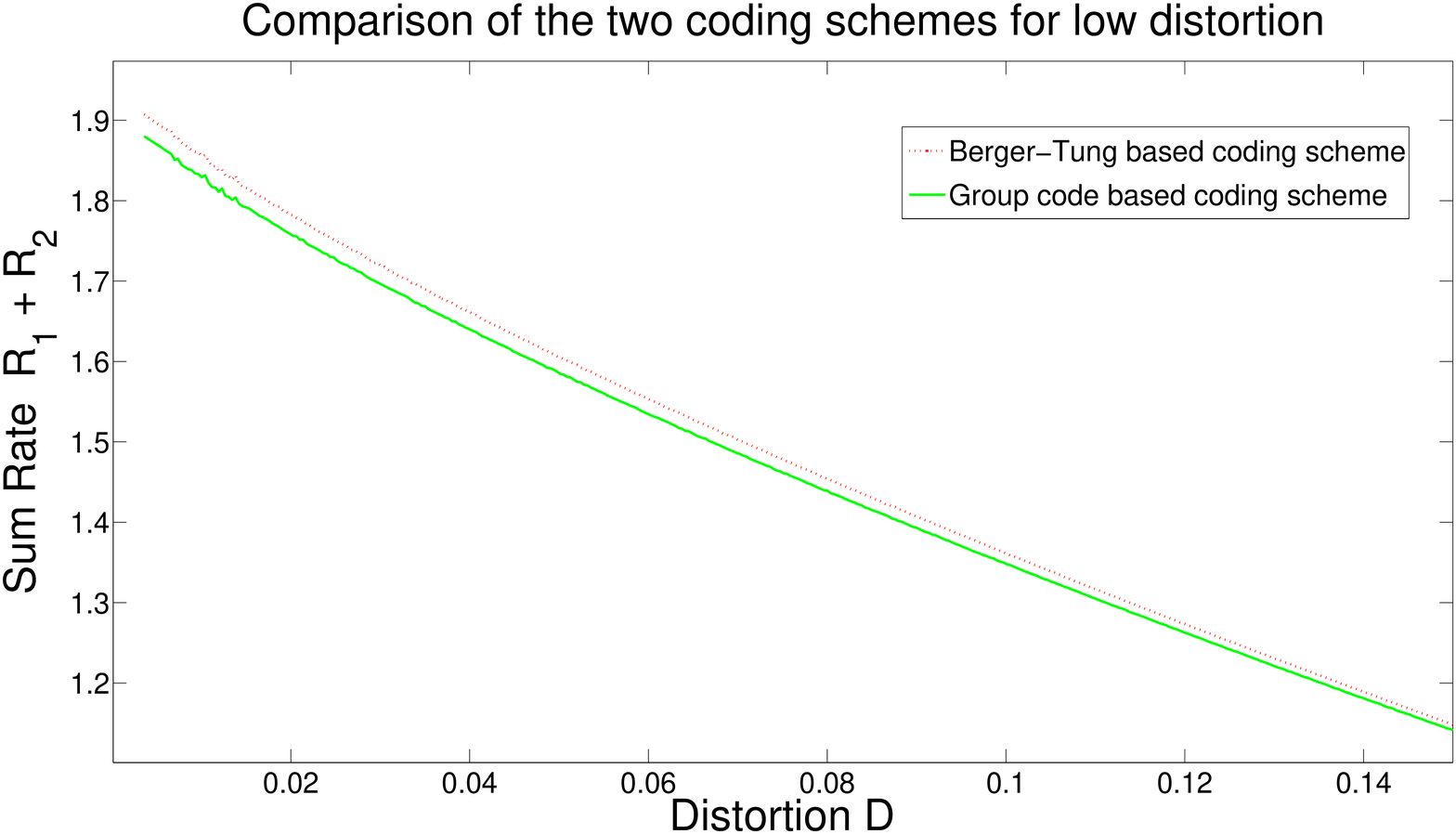, height = 2.5in, width = 0.49 \textwidth}\label{fig:noconv_zoom}
}
\hspace{-0.45 in}
\subfigure{
\epsfig{figure=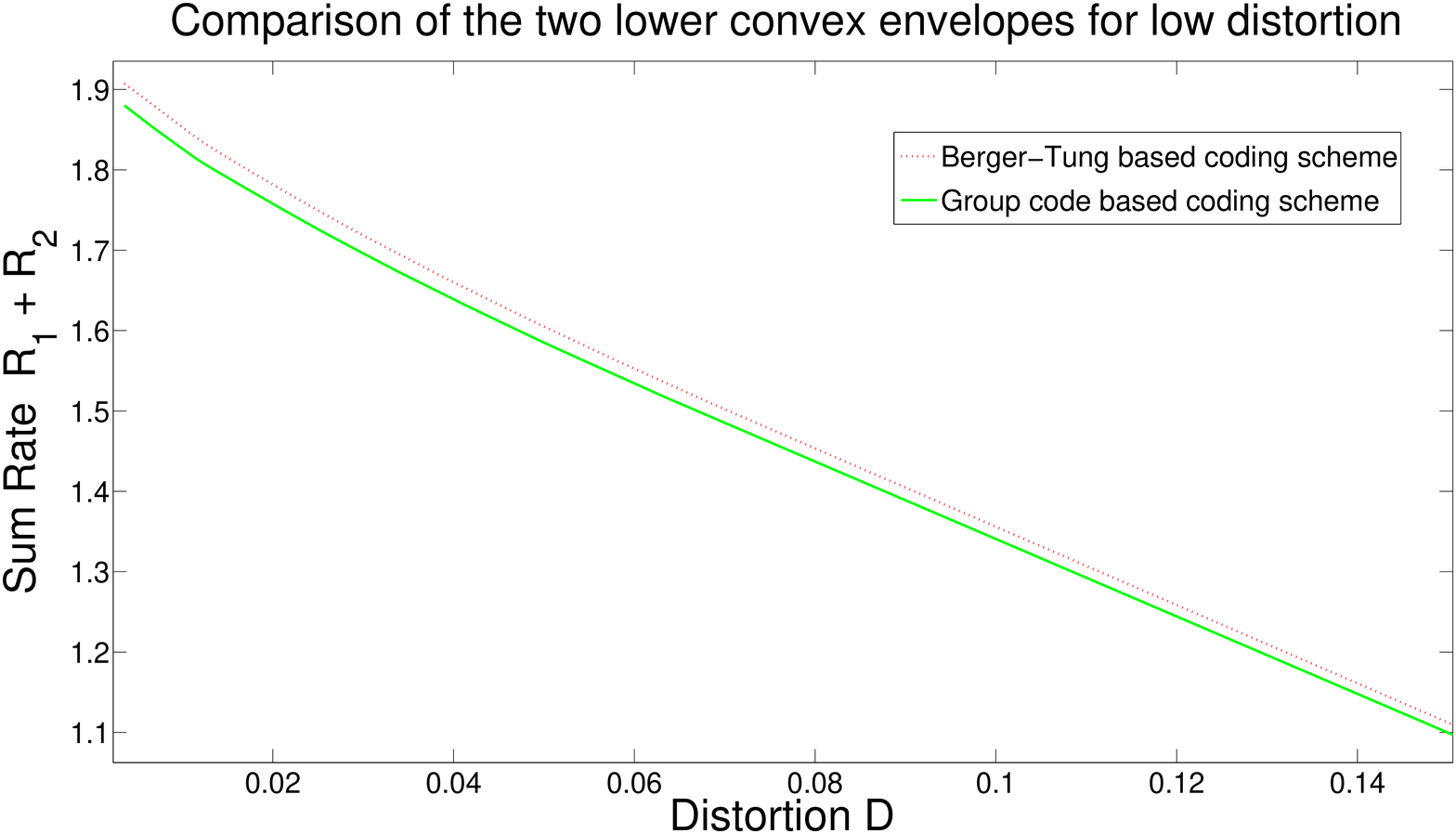, height = 2.5in, width = 0.49 \textwidth} \label{fig:conv_zoom}
}
\caption{Zoomed versions of Figures \ref{fig:sumrate_noconv} and \ref{fig:sumrate_convhull}}
\label{fig:zoomedfigs}
\end{figure}

Motivation of choosing this example is as follows. Evaluation of the Berger-tung rate region is a computationally intensive operation since it involves solving a nonconvex optimization problem. The only procedure that we are aware of for this is  using  linear programming followed by quantizing the probability space and searching for optimum values \cite{gu-jana-effros}. The computational complexity increases dramatically as the size of the alphabet of the sources goes beyond two. Hence we chose the simplest nontrivial lossy example to make the point. But then we have to deal with the fact that there are only 3 abelian groups of order less than or equal to 4. One of the three groups corresponds to the Berger-Tung bound.  We would like to remark that even for this simple example, the Berger-Tung bound is not tight. We expect the gains afforded by Theorem \ref{thm:mainthm} over the rate region of Lemma \ref{lemma:BTschemelemma} would increase as we increase the cardinality of the source alphabets.

\begin{table} \begin{center}
\begin{tabular}{|c||c|c|}
\hline
$P_{XY}$ & $0$ & $1$ \\
\hline \hline
$0$ & $0.3381$ & $0.1494$ \\
$1$ & $0.2291$ & $0.2834$ \\
\hline
\end{tabular}
\end{center}
\caption{Joint distribution used for example in Figures \ref{fig:sumrate_noconv},\ref{fig:sumrate_convhull}} \label{table:BT_example1}
\end{table}

\section{Conclusion} \label{sec:conclusions}
We have introduced structured codes built over arbitrary abelian groups for lossless and lossy source coding and derived their performance limits. We also derived a coding theorem based on nested group codes for reconstructing an arbitrary function of the sources based on a fidelity criterion. The encoding proceeds sequentially in stages based on the primary cyclic decomposition of the underlying abelian group. This coding scheme recovers the known rate regions of many distributed source coding problems while presenting new rate regions to others. The usefulness of the scheme is demonstrated with both lossless and lossy examples.

\section*{Acknowledgements} The authors would like to thank
Professor Hans-Andrea Loeliger of ETH, Zurich and Dr.~Soumya Jana of University of Illinois, Urbana-Champaign for helpful discussions.

\section*{Appendix}
\appendix
\section{Good Group Channel Codes} \label{sec:goodchproof}
We prove the existence of channel codes built over the space $\ringr{n}$ which are good for the triple $(\mc{Z}, \mc{S}, P_{ZS})$ according to Definition \ref{defi:goodchcodedefi}. Recall that the group $\ring{p^r}$ has $(r-1)$ non-trivial subgroups, namely $p^i \ring{p^r}, 1 \leq i \leq r-1$. Let the random variable $Z$ take values from the group $\ringr{}$, i.e., $\mc{Z} = \ringr{}$ and further let it be non-redundant. Let Hom$(\ring{p^r}^n, \ring{p^r}^k)$ be the set of all homomorphisms from $\ring{p^r}^n$ to $\ring{p^r}^k$. Let $\phi(\cdot)$ be a homomorphism picked at random with uniform probability from Hom$(\ringr{n},\ringr{k})$.

We start by proving a couple of lemmas.
\begin{lemma} \label{lemma:uniflem}
For a homomorphism $\phi(\cdot)$ randomly chosen from Hom$(\ringr{n},\ringr{k})$, the probability that a given sequence $z^n$ belongs to $\ker(\phi)$ in $\ring{p^r}^n$ depends on which subgroup of $\ringr{}$ the sequence $z^n$ belongs to. Specifically
\begin{equation}
P(\phi(z^n) = 0^k) = \left\{ \begin{array}{cl} p^{-(r-i)k} & \mbox{if } z^n \in p^i \ring{p^r}^n \backslash \, p^{i+1} \ring{p^r}^n, 0 \leq i < r \\
1 & \mbox{if }z^n \in p^r \ring{p^r}^n \end{array} \right.
\end{equation}
\end{lemma}

\begin{proof}[\textbf{Proof}:]
Clearly, $z^n \in p^r \ring{p^r}^n$ implies $z^n = 0^n$ \footnote{If we consider homomorphisms from $\ringr{n}$ to $\ring{m}^k$ for an arbitrary integer $m$, all such homomorphisms have $d\ringr{n}$ as their kernel where $d = (p^r,m)$ is the greatest common divisor of $p^r$ and $m$. Unless $d = p^r$, there would be exponentially many $z^n$ for which $P(\phi(z^n) = 0) = 1$ for all $\phi \in \mbox{Hom}(\ring{p^r}^n,\ring{m}^k)$ and this results in bad channel codes (see equation (\ref{eq:footnoteeq})). Thus, $p^r$ has to divide $m$ and all such $m$ give identical performances as $m=p^r$.}.  In this case, the probability of the event $\{\phi(z^n) = 0\}$ is $1$.

Let the $k \times n$ matrix $\Phi$ be the matrix representation of the homomorphism $\phi(\cdot)$. Let the first row of $\Phi$ be $(\alpha_1, \dots, \alpha_n)$. Consider $\phi_1 \colon \ringr{n} \rightarrow \ringr{}$, the homomorphism corresponding to the first row of $\Phi$. The total number of possibilities for $\phi_1(\cdot)$ is $(p^r)^n$.

Let us consider the case where $z^n \in \ring{p^r}^n \backslash p \ring{p^r}^n$. In this case, $z^n$ contains at least one element, say $z_i$ which is invertible in $\ring{p^r}$. Let us count the number of homomorphisms $\phi(\cdot)$ that map such a sequence $z^n$ to a given $c \in \ringr{k}$. We need to choose the $k$ homomorphisms $\phi_i(\cdot), 1 \leq i \leq k$ such that $\phi_i(z^n) = c_i$ for $1 \leq i \leq k$. Let us count the number of homomorphisms $\phi_1(\cdot)$ that map $z^n$ to $c_1$. In this case, we can choose $\alpha_j, j \neq i$ to be arbitrary and fix $\alpha_i$ as
\begin{equation}
\alpha_i = z_i^{-1}\left( c_1 - \sum_{\stackrel{j=1}{j \neq i}}^n \alpha_j z_j \right)
\end{equation}
Thus the number of favorable homomorphisms $\phi_1(\cdot)$ is $(p^r)^{(n-1)}$. Thus the probability that a randomly chosen homomorphism $\phi_1(\cdot)$ maps $z^n$ to $c_1$ is $p^{-r}$. Since each of the $k$ homomorphisms $\phi_i$ can be chosen independently, we have
\begin{equation} \label{eq:hm0groupeq}
P( \phi(z^n) = c) = p^{-rk} \quad \mbox{if } \, z^n \in \ring{p^r}^n \backslash p \ring{p^r}^n
\end{equation}
Putting $c = 0^k$ in equation (\ref{eq:hm0groupeq}), we see that the claim in Lemma \ref{lemma:uniflem} is valid for $z^n \in \ringr{n} \backslash p \ringr{n}$. Now, consider $z^n \in p^i \ringr{n} \backslash p^{i+1} \ringr{n}$ for a general $0 < i < r$. Any such $z^n$ can be written as $p^i \tilde{z}^n$ for $\tilde{z}^n \in \ringr{n} \backslash p \ringr{n}$. Thus, the event $\{ \phi(z^n) = 0 \}$ will be true if and only if $\{\phi(\tilde{z}^n) = t\}$ for some $t \in p^{r-i} \ringr{k}$. Hence,
\begin{align}
P(\phi(z^n) = 0) &= P \left( \bigcup_{t \in p^{r-i} \ringr{k}} (\phi(\tilde{z}^n) = t) \right) \\
&= \sum_{t \in p^{r-i} \ringr{k}} P( \phi(\tilde{z}^n) = t) \\
&= | p^{r-i} \ringr{k} | p^{-rk} \\
&= p^{-(r-i)k} \label{eq:hmigroupeq}
\end{align}
This proves the claim of Lemma \ref{lemma:uniflem}.
\end{proof}

We now estimate the size of the intersection of the conditionally typical set $\typset{}{s^n}$ with cosets of $p^i \ring{p^r}^n$ in $\ring{p^r}^n$.
\begin{lemma} \label{lemma:sizelem}
For a given $z^n \in \typset{}{s^n}$, consider $(z^n + p^i \ring{p^r}^n)$, the coset of $p^i \ring{p^r}^n$ in $\ring{p^r}^n$. Define the set $S_{i,\epsilon}(z^n,s^n)$ as $S_{i,\epsilon}(z^n,s^n) \triangleq (z^n + p^i \ring{p^r}^n) \cap \typset{}{s^n}$. A uniform bound on the cardinality of this set is given by
\begin{equation}
\frac{1}{n} \log |S_{i,\epsilon}| \leq H(Z|S) - H([Z]_i|S) + \delta(\epsilon) \qquad \mbox{for }0 \leq i \leq r
\end{equation}
where $\delta(\epsilon) \rightarrow 0$ as $\epsilon \rightarrow 0$. The random variable $[Z]_i$ is defined in the following manner: It takes values from the set of all distinct cosets of $p^i \ring{p^r}$ in $\ring{p^r}$. The probability that $[Z]_i$ takes a particular coset as its value is equal to the sum of the probabilities of the elements forming that coset.
\begin{equation}
P([Z]_i = a + p^i \ring{p^r} \mid S = s) = \sum_{z \in a + p^i \ring{p^r}} P_{Z|S}(z \mid s) \quad \forall s \in \mc{S}.
\end{equation}
We have the nesting relation $S_{i+1,\epsilon}(z^n,s^n) \subset S_{i,\epsilon}(z^n,s^n)$ for $0 \leq i \leq r-1$. However, each nested set is exponentially smaller in size since $H([Z]_i)$ increases monotonically with $i$. Thus, with the same definitions as above, we also have that
\begin{equation}
\frac{1}{n} \log \left( |S_{i,\epsilon}(z^n,s^n)| - |S_{i+1,\epsilon}(z^n,s^n)| \right) \leq H(Z|S) - H([Z]_i|S) + \delta_1 (\epsilon) \qquad \mbox{for }0 \leq i \leq r-1
\end{equation}
where $\delta_1(\epsilon) \rightarrow 0$ as $\epsilon \rightarrow 0$.
\end{lemma}

\begin{proof}[\textbf{Proof}:] The set $S_{i,\epsilon}(z^n,s^n)$ can   be thought of as all those sequences $\tilde{z}^n \in  \typset{}{s^n}$ such that the difference $w^n \triangleq \tilde{z}^n  - z^n \in p^i \ringr{}$. Let $W$ be a random variable taking values  in $p^i \ringr{}$ and jointly distributed with $(Z,S)$ according to  $P_{W|ZS}$. Define the random variable $\tilde{Z} \triangleq Z +  W$. Let $P_{W|ZS}$ be such that $P_{\tilde{ZS}} = P_{ZS}$. Then, for  a given distribution $P_{W|ZS}$, every sequence $\tilde{z}^n$ that  belongs to the set of conditionally typical sequences given  $(z^n,s^n)$ will belong to the set  $S_{i,\epsilon}(z^n,s^n)$. Conversely, following the type counting  lemma and the continuity of entropy as a function of probability  distributions \cite{csiszarbook},  every sequence $\tilde{z}^n  \in S_{i,\epsilon}(z^n,s^n)$ belongs to the set of conditionally  typical sequences given $(z^n,s^n)$ for some such joint distribution  $P_{W|ZS}$. Thus estimating the  size of the set $S_{i,\epsilon}(z^n,s^n)$ reduces to estimating the  maximum of $H(\tilde{Z} \mid Z,S)$, or equivalently the maximum of  $H(Z,W \mid S)$ over all joint distributions $P_{ZSW}$ such that  $P_{(Z+W),S} = P_{ZS}$.

We formulate this problem as a convex optimization problem in the following manner. Recall that the alphabet of $Z$ is the group $\ringr{}$. Hence, $H(Z,W \mid S)$ is a concave function of the $|\mc{Z}| |\mc{S}| |p^i \ringr{}|$ variables $P_{ZSW}(Z = z, S = s, W = w), z \in \mc{Z}, S \in \mc{S}, w \in p^i \ringr{}$ and maximizing this conditional entropy is a convex minimization problem. Since the distribution $P_{ZS}$ is fixed, these variables satisfy the marginal constraint
\begin{equation} \label{eq:margeconstraint}
\sum_{w \in p^i \ringr{}} P_{ZW|S}(Z = z, W = w \mid S = s) = P_{Z|S}(z \mid s) \quad \forall z \in \mc{Z}, s \in \mc{S}
\end{equation}
The other constraint to be satisfied is that the random variable $\tilde{Z} = Z + W$ is jointly distributed with $S$ in the same way as $Z$, i.e., $P_{\tilde{Z}S} = P_{ZS}$. This can be expressed as
\begin{equation} \label{eq:samedistconstraint}
\sum_{w \in p^i \ringr{}} P_{ZW|S}(Z = z-w, W = w \mid S = s) = P_{Z|S}(z \mid s) \quad \forall z \in \mc{Z}, s \in \mc{S}.
\end{equation}

Thus the convex optimization problem can be stated as
\begin{align}
\nonumber \mbox{minimize} \,\, &- H(Z, W \mid S) \\
\nonumber \mbox{subject to} & \,\, \sum_{w \in p^i \ringr{}} P_{ZW|S}(z,w \mid s) = P_{Z|S}(z \mid s) \quad \forall z \in \mc{Z}, s \in \mc{S}, \\
& \, \, \sum_{w \in p^i \ringr{}} P_{ZW|S}(z-w,w \mid s) = P_{Z|S}(z \mid s) \quad \forall z \in \mc{Z}, s \in \mc{S}. \label{eq:convexprob}
\end{align}
Note that the objective function to be minimized is convex and the constraints of equations (\ref{eq:margeconstraint}) and (\ref{eq:samedistconstraint}) on $P_{ZW|S}(Z = z, W = w \mid S = s)$ are affine. Thus, the Karush-Kuhn-Tucker (KKT) conditions \cite{boyd} are necessary and sufficient for the points to be primal and dual optimal. We now derive the KKT conditions for this problem. We formulate the dual problem as
\begin{align}
\nonumber D(P_{ZW|S}) &= \, - \sum_{s \in \mc{S}} P_S(s) \left( \sum_{z \in \mc{Z}, w \in p^i \ringr{}} P_{ZW|S}(z,w \mid s) \log \frac{1}{P_{ZW|S}(z,w \mid s)} \right) \\ \nonumber &+ \sum_{z \in \mc{Z}, s \in \mc{S}} \lambda_{z,s} \left( \sum_{w \in p^i \ringr{}} P_{ZW|S}(z-w,w \mid s) - P_{Z|S}(z \mid s) \right)  \\ &+ \sum_{z \in \mc{Z}, s \in \mc{S}} \gamma_{z,s} \left( \sum_{w \in p^i \ringr{}} P_{ZW|S}(z,w \mid s) - P_{Z|S}(z \mid s) \right)
\end{align}
where $\{ \lambda_{z,s} \}, \{ \gamma_{z,s} \}$ are the Lagrange multipliers. Differentiating with respect to $P_{ZW|S}(Z = z, W = w \mid S = s)$ and setting the derivative to $0$, we get
\begin{align}
\frac{\partial D(P_{ZW|S})}{\partial P_{ZW|S}(z,w \mid s)} &=  P_S(s)( 1 + \log P_{ZW|S}(z,w \mid s) ) + \lambda_{(z + w),s} + \gamma_{z,s} = 0 \\
\implies \lambda_{(z+w),s} + \gamma_{z,s} &= -P_S(s)(1 + \log P_{ZW|S}(z,w \mid s)) \quad \forall z \in \mc{Z}, s \in \mc{S}, w \in p^i \ringr{}.
\end{align}
Summing over all $z \in p^i \ringr{}$ for a given $s \in \mc{S}$, we see that for all $w \in p^i \ringr{}$, the summation $\sum_{z \in p^i \ringr{}} \left( \lambda_{(z+w),s} + \gamma_{z,s} \right)$ is the same. This implies that
\begin{equation} \label{eq:kkteqs}
\prod_{z \in p^i \ringr{}} P_{ZW|S}(z,w \mid s) = \mbox{ constant} \quad \forall w \in p^i \ringr{}, \, \forall s \in \mc{S}.
\end{equation}
These $ |\mc{S}| p^{r-i}$ equations form the KKT equations and any solution that satisfies equations (\ref{eq:margeconstraint}), (\ref{eq:samedistconstraint}) and (\ref{eq:kkteqs}) is the optimal solution to the optimization problem (\ref{eq:convexprob}). We claim that the solution to this system of equations is given by
\begin{equation} \label{eq:pxwchoice}
P_{ZW|S}(z,w \mid s) = \frac{P_{Z|S}(z \mid s) P_{Z|S}(z+w \mid s)}{P_{Z | S}(z + p^i \ringr{} \mid s)}
\end{equation}
For this choice of $P_{ZW|S}(z,w \mid s)$, we now show that equation (\ref{eq:margeconstraint}) is satisfied.
\begin{align}
\sum_{w \in p^i \ringr{}} P_{ZW|S}(z,w \mid s) &= \sum_{w \in p^i \ringr{}} \frac{P_{Z|S}(z \mid s) P_{Z|S}(z+w \mid s)}{P_{Z|S}(z + p^i \ringr{} \mid s)} \\
&= \frac{P_{Z|S}(z \mid s)}{P_{Z|S}(z+p^i \ringr{} \mid s)} \sum_{w \in p^i \ringr{}} P_{Z|S}(z+w \mid s) \\
&= P_{Z|S}(z \mid s) \quad \forall z \in \mc{Z}, s \in \mc{S}.
\end{align}
Next, lets show that the choice of $P_{ZW|S}(z,w \mid s)$ in equation (\ref{eq:pxwchoice}) satisfies equation (\ref{eq:samedistconstraint}).
\begin{align}
\sum_{w \in p^i \ringr{}} P_{ZW|S}(Z = z-w, W = w \mid S = s) &= \sum_{w \in p^i \ringr{}} \frac{P_{Z|S}(z-w \mid s) P_{Z|S}(z \mid s)}{P_{Z|S}(z-w + p^i \ringr{} \mid s)} \\
&= P_{Z|S}(z \mid s) \sum_{w \in p^i \ringr{}} \frac{P_{Z|S}(z-w \mid s)}{P_Z(z + p^i \ringr{} \mid s)} \\ &= P_{Z|S}(z \mid s) \quad \forall z \in \mc{Z}, s \in \mc{S}.
\end{align}
Finally, we show that this choice of $P_{ZW|S}(z,w \mid s)$ satisfies the KKT conditions given by equation (\ref{eq:kkteqs}).
\begin{align}
\prod_{z \in p^i \ringr{}} P_{ZW|S}(z,w \mid s) &= \prod_{z \in p^i \ringr{}} \frac{P_{Z|S}(z \mid s) P_{Z|S}(z+w \mid s)}{P_{Z|S}(z + p^i \ringr{} \mid s)} \\
&= \left(\frac{1}{P_{Z|S}(p^i \ringr{} \mid s)}\right)^{p^{r-i}} \prod_{z \in p^i \ringr{}} P_{Z|S}^2(z \mid s)
\end{align}
which is independent of $w$ and is the same for any $w \in p^i \ringr{}$. Thus, equation (\ref{eq:pxwchoice}) indeed is the solution to the optimization problem described by equation (\ref{eq:convexprob}). Let us now compute the maximum value that the entropy $H(W \mid Z,S)$ takes for this choice of the conditional distribution $P_{ZW|S}$.
\begin{align}
H(W \mid Z,S) &= \sum_{s \in \mc{S}} \sum_{z \in \mc{Z}} P_{ZS}(z,s) \left( \sum_{w \in p^i \ringr{}} P_{W|ZS}(w \mid z,s) \log \frac{1}{P_{W|ZS}(w \mid z,s)} \right) \\
&= \sum_{s \in \mc{S}} \sum_{z \in \mc{Z}} P_{ZS}(z,s) \left( \sum_{w \in p^i \ringr{}} \frac{P_{Z|S}(z+w \mid s)}{P_{Z|S}(z + p^i \ringr{} \mid s)} \log \frac{P_{Z|S}(z + p^i \ringr{} \mid s)}{P_{Z|S}(z+w \mid s)} \right) \label{eq:tempevaleq}
\end{align}
Let $\mc{DC}$ be the set of all distinct cosets of $p^i \ringr{}$ in $\ringr{}$ and let $\mc{DC}(z)$ be the unique set in $\mc{DC}$ that contains $z$.
Let us evaluate the summation in the brackets of equation (\ref{eq:tempevaleq}) first.
\begin{align}
\sum_{w \in p^i \ringr{}} \frac{P_{Z|S}(z+w \mid s)}{P_{Z|S}(z + p^i \ringr{} \mid s)} \log \frac{P_{Z|S}(z + p^i \ringr{} \mid s)}{P_{Z|S}(z+w \mid s)} &= \sum_{w \in p^i \ringr{}} \frac{P_{Z|S}(z+w \mid s)}{P_{Z|S}( \mc{DC}(z) \mid s)} \log \frac{P_{Z|S}(\mc{DC}(z) \mid s)}{P_{Z|S}(z+w \mid s)} \\ &= \log P_{Z|S}(\mc{DC}(z) \mid s) + \frac{\sum_{z^{'} \in \mc{DC}(z)} P_{Z|S}(z^{'} \mid s) \log \frac{1}{P_{Z|S}(z^{'} \mid s)}}{P_{Z|S}(\mc{DC}(z) \mid s)}
\end{align}
This sum is dependent on $z$ only through the coset $\mc{DC}(z)$ to which $z$ belongs. Thus, the sum is the same for any two $z$ that belong to the same coset of $p^i \ringr{}$ in $\ringr{}$. Thus, we have
\begin{align}
H(W \mid Z,S) &= \sum_{s \in \mc{S}} P_S(s) \sum_{T \in \mc{DC}} \sum_{z \in T} P_{Z|S}(z \mid s) \left( \log P_{Z|S}(T \mid s) + \frac{\sum_{z^{'} \in T} P_{Z|S}(z^{'} \mid s) \log \frac{1}{P_{Z|S}(z^{'} \mid s)}}{P_{Z|S}(T \mid s)} \right) \\
&= \sum_{s \in \mc{S}} \sum_{T \in \mc{DC}} P_{Z|S}(T \mid s) \left( \log P_{Z|S}(T \mid s) + \frac{\sum_{z^{'} \in T} P_{Z|S}(z^{'} \mid s) \log \frac{1}{P_{Z|S}(z^{'} \mid s)}}{P_{Z|S}(T \mid s)} \right) \\
&= \sum_{s \in S} \sum_{z^{'} \in \mc{Z}} P_{Z|S}(z^{'} \mid s) \log \frac{1}{P_{Z|S}(z^{'} \mid s)} + \sum_{s \in S} P_{Z|S}(T \mid s) \log \frac{1}{P_{Z|S}(T \mid s)} \\
&= H(Z \mid S) - H([Z]_i \mid S)
\end{align}
where $[Z]_i$ is as defined in Lemma \ref{lemma:sizelem}.
\end{proof}

We are now ready to prove the existence of good group channel codes. Let $Z$ take values in the group $\ringr{}$ and further be non-redundant. Coding is done in blocks of length $n$. We show the existence of a good channel code by averaging the probability of a decoding error over all possible choices of the homomorphism $\phi(\cdot)$ from the family Hom$(\ringr{n},\ringr{k})$. Let $H$ be the parity check matrix and $\mc{C}$ be the kernel of a randomly chosen homomorphism $\phi(\cdot)$.

The probability of the set $B_{\epsilon}(\mc{C})$  can be written as
\begin{align}
P_{ZS}(B_{\epsilon}(\mc{C})) &= \sum_{(z^n,s^n)} P_{ZS}(z^n,s^n)  I\left( \bigcup_{\stackrel{(\tilde{z}^n,s^n) \in \typset{}{Z,S}}{\tilde{z}^n \neq z^n}} (\phi(\tilde{z}^n) = \phi(z^n)) \right)
\end{align}
where $I(E)$ is the indicator of the event $E$. Taking the expectation of this probability, we get
\begin{align}
\Bbb{E}(P_{ZS}(B_{\epsilon}(\mc{C}))) &= \sum_{(z^n,s^n)} P_{ZS}(z^n,s^n) P \left( \bigcup_{\stackrel{(\tilde{z}^n,s^n) \in \typset{}{Z,S}}{\tilde{z}^n \neq z^n}} (\phi(\tilde{z}^n) = \phi(z^n)) \right) \\
&\leq \sum_{(z^n,s^n) \in \typset{}{Z,S}} P_{ZS}(z^n,s^n) P \left( \bigcup_{\stackrel{(\tilde{z}^n,s^n) \in \typset{}{Z,S}}{\tilde{z}^n \neq z^n}} (\phi(\tilde{z}^n) = \phi(z^n)) \right) + \sum_{(z^n,s^n) \notin \typset{}{Z,S}} P_{ZS}(z^n,s^n) \\
&\leq \sum_{(z^n,s^n) \in \typset{}{Z,S}} P_{ZS}(z^n,s^n) P \left( \bigcup_{\stackrel{(\tilde{z}^n,s^n) \in \typset{}{Z,S}}{\tilde{z}^n \neq z^n}} (\phi(\tilde{z}^n - z^n) = 0^k) \right) \,+ \, \delta_1 \label{eq:Bepsprob}
\end{align}
where $\delta_1 \rightarrow 0$ as $n \rightarrow \infty$.

We now derive a uniform bound for the probability that for a given $(z^n,s^n) \in \typset{}{Z,S}$, a randomly chosen homomorphism maps $\tilde{z}^n$ to the same syndrome as $z^n$ for some $\tilde{z}^n$ such that $(\tilde{z}^n,s^n) \in \typset{}{Z,S}$. From Lemma \ref{lemma:uniflem} and \ref{lemma:sizelem}, we see that this probability depends on which of the sets $S_{i,\epsilon}(z^n,s^n), 0 \leq i < r$ the sequence $\tilde{z}^n$ belongs to.
\begin{align}
P \left( \bigcup_{\stackrel{(\tilde{z}^n,s^n) \in \typset{}{Z,S}}{\tilde{z}^n \neq z^n}} \phi(\tilde{z}^n - z^n) = 0^k\right) &\leq \sum_{\stackrel{(\tilde{z}^n,s^n) \in \typset{}{Z,S}}{\tilde{z}^n \neq z^n}} P(\phi(\tilde{z}^n - z^n) = 0^k) \label{eq:footnoteeq} \\
&= \sum_{i=0}^{r-1} \sum_{\tilde{z}^n \in S_{i,\epsilon}(z^n,s^n) \backslash S_{i+1,\epsilon}(z^n,s^n)} P(\phi(\tilde{z}^n - z^n) = 0^k) \\
&\stackrel{(a)}{=} \sum_{i=0}^{r-1} \frac{|S_{i,\epsilon}(z^n,s^n)| - |S_{i+1,\epsilon}(z^n,s^n)|}{p^{(r-i)k}} \\
&\stackrel{(b)}{\leq} \sum_{i=0}^{r-1} \exp_2 \left( n \left[ H(Z|S) - H([Z]_i|S) - \frac{k}{n} (r-i) \log p + \delta_2(\epsilon) \right] \right) \label{eq:channelerroreq}
\end{align}
where $(a)$ follows from Lemma \ref{lemma:uniflem} and $(b)$ follows from Lemma \ref{lemma:sizelem}. If this summation were to go to zero with block length, it would follow from equation (\ref{eq:Bepsprob}) that the expected probability of the set $B_{\epsilon}(\mc{C})$ also goes to zero. This implies the existence of at least one homomorphism $\phi(\cdot)$ such that the associated codebook $\mc{C}$ satisfies for a given $\epsilon > 0$, $P_{ZS}(B_{\epsilon}(\mc{C})) \leq \epsilon$ for sufficiently large block length.

The summation in equation (\ref{eq:channelerroreq}) goes to zero if each of the terms goes to zero. This happens if
\begin{equation}
\frac{k}{n} \frac{r-i}{r} \log p^r \geq H(Z|S) - H([Z]_i|S) + \delta_2(\epsilon) \quad \mbox{for } i = 0,\dots, r-1
\end{equation}
or equivalently
\begin{equation}
\frac{k(n)}{n} \log p^r > \max_{0 \leq i < r} \left( \frac{r}{r-i} \right) (H(Z|S) - H([Z]_i|S)) + \delta_2(\epsilon)
\end{equation}
It is clear that in the limit as $n \rightarrow \infty$, good group channel codes exist such that the dimensions of the associated parity check matrices satisfy equation (\ref{eq:goodlinchannellimiteq}). When $\mc{C}$ is a good channel code, define the decoding function $\psi \colon \ringr{k} \times \mc{S}^n \rightarrow \ringr{n}$ for a given $(z^n,s^n)$ as the unique member of the set $\{\hat{z}^n \colon H\hat{z}^n = Hz^n, (\hat{z}^n,s^n) \in \typset{}{Z,S} \}$.

\section{Good Group Source Codes} \label{sec:goodsrcproof}
We prove the existence of source codes built over the space $\ringr{n}$ which are good for the triple $(\mc{X},\mc{U},P_{XU})$ according to Definition \ref{defi:goodsrccodedefi}. Let the random variable $U$ take values from the group $\ringr{}$, i.e., $\mc{U} = \ringr{}$ and let $U$ be non-redundant. Let $\phi \colon \ring{p^r}^n \rightarrow \ring{p^r}^k$ be a homomorphism for some $k$ to be fixed later. The codebook $\mc{C}$ is the kernel $\ker(\phi)$ of this homomorphism. Note that $\ker(\phi) < \ring{p^r}^n$ and hence the codebook has a group structure. We show the existence of a good code $\mc{C}$ by averaging the probability of error over all possible choices of $\phi(\cdot)$ from the family of all homomorphisms Hom$(\ring{p^r}^n, \ring{p^r}^k)$.

Recall the definition of the set $A_{\epsilon}(\mc{C})$ from equation (\ref{eq:Asetdefi}). The probability of this set can be written as
\begin{align}
P(A_{\epsilon}(\mc{C})) &= \sum_{x^n} P_X(x^n) I \left( \bigcup_{u^n \in \mc{C}} (x^n,u^n) \in \typset{}{X,U} \right)
\end{align}
The expected value of this probability is
\begin{align}
\Bbb{E}(P(A_{\epsilon}(\mc{C}))) &= \sum_{x^n} P_X(x^n) P \left( \bigcup_{u^n \in \mc{C}} (x^n,u^n) \in \typset{}{X,U} \right) \\
&\geq \sum_{x^n \in \typset{}{X}} P_X(x^n) P \left( \bigcup_{u^n \in \mc{C}} (x^n,u^n) \in \typset{}{X,U} \right) \label{eq:Aepsprob}
\end{align}

For a typical $x^n$, let us compute the probability that there exists no $u^n \in \mc{C}$ jointly typical with the source sequence $x^n$. Define the random variable $\Theta(x^n)$ as
\begin{equation} \label{eq:Thetadef}
\Theta(x^n) = \sum_{u^n \in \typset{}{x^n}} 1_{\{u^n \in \mc{C} \}}.
\end{equation}
$\Theta(x^n)$ counts the number of $u^n$ sequences in the codebook $\mc{C}$ that are jointly typical with $x^n$. The error event $E$ given that the source sequence is $x^n$ is equivalent to the event $\{\Theta(x^n) = 0\}$. Thus, we need to evaluate the probability of this event. Note that $\Theta(x^n)$ is a sum of indicator random variables some of which might be dependent. This dependence arises from the structural constraint on the codebook $\mc{C}$. For example, $u_1^n \in \mc{C}$ implies that $k u_1^n \in \mc{C}$ as well for any $k \in \ringr{}$. We use Suen's inequality \cite{janson} to bound this probability.

In order to use Suen's inequality, we need to form the dependency graph between these indicator random variables. We do this in a series of lemmas. We first evaluate the probability that a given typical sequence belongs to the kernel of a randomly chosen homomorphism. Since $U$ is assumed to be non-redundant, by Lemma \ref{lemma:uniflem}, we have
\begin{equation} \label{eq:Iiprob}
P(u^n \in \mc{C}) = p^{-rk}
\end{equation}

We now turn our attention to pairwise relations between the indicator random variables. For two $n$-length sequences $u_1^n,u_2^n$, define the matrices $M_{k,l}(u_1^n,u_2^n), 1 \leq k,l \leq n$ and $k \neq l$ as
\begin{equation}
M_{k,l}(u_1^n,u_2^n) = \left[ \begin{array}{cc} u_{1k} & u_{1l} \\ u_{2k} & u_{2l} \end{array} \right]
\end{equation}
Let $m_{k,l}(u_1^n,u_2^n)$ be the determinant of the matrix $M_{k,l}(u_1^n,u_2^n)$. Define the set
\begin{equation}
M(u_1^n,u_2^n) \triangleq \{ m_{k,l}(u_1^n,u_2^n) \colon u_{1k}^{-1} \mbox{ exists} \}
\end{equation}
Note that the set $M(u_1^n,u_2^n)$ is non-empty since $u_1^n$ is assumed to be a non-redundant sequence. Let $D(u_1^n,u_2^n)$ be the smallest subgroup of $\ringr{}$ that contains the set $M(u_1^n,u_2^n)$. As will be shown, the probability that both $u_1^n$ and $u_2^n$ belong to the kernel of a randomly chosen homomorphism depends on $D(u_1^n,u_2^n)$. For ease of notation, we suppress the dependence of the various quantities on the sequences $u_1^n,u_2^n$ in what follows.

\begin{lemma} \label{lemma:jointhomoprob}
For two non-redundant sequences $u_1^n, u_2^n$, the probability that a random homomorphism $\phi \colon \ringr{n} \rightarrow \ringr{k}$ maps the sequences to $0^k$ is
\begin{equation}
P(\phi(u_1^n) = \phi(u_2^n) = 0^k) = p^{-(2r-i)k} \quad \mbox{if }\, D(u_1^n,u_2^n) = p^i \ringr{}, 0 \leq i \leq r
\end{equation}
\end{lemma}

\begin{proof}[\textbf{Proof}:]
Let the homomorphism $\phi(\cdot)$ be decomposed as $\phi_i \colon \ringr{n} \rightarrow \ringr{}, 1 \leq i \leq k$. We first count the number of homomorphisms $\phi_1(\cdot)$ that map both $u_1^n$ and $u_2^n$ to $0$. Recall that $\phi_1(u_1^n)$ can be expressed as the linear combination $\phi_1(u_1^n) = \sum_{j=1}^n \alpha_j u_{1j}$ for $\alpha_j \in \ringr{}, 1 \leq j \leq n$. Thus, we need to find the number of solutions $\{ \alpha_j\}_{j=1}^{n}$ that simultaneously satisfy the equations
\begin{align}
\sum_{j=1}^{n} \alpha_j u_{1j} &= 0 \label{eq:alpeq1}\\
\sum_{j=1}^{n} \alpha_j u_{2j} &= 0 \label{eq:alpeq2}
\end{align}

If $D(u_1^n,u_2^n) = p^i \ringr{}$, then there exists some $1 \leq k \leq n$ such that $u_{1k}^{-1}$ exists and $m_{k,j^*} \in p^i \ringr{} \backslash p^{i+1} \ringr{}$ for some $1 \leq j^* \leq n, j^* \neq k$. Fix such a $k$. Then, any solution to the equation (\ref{eq:alpeq1}) must be of the form $\alpha_j, j \neq k \mbox{ arbitrary}, \alpha_k = -u_{1k}^{-1} \sum_{j \neq k} \alpha_j u_{1j}$ for some $k$ such that $u_{1k}^{-1}$ exists. Thus, the total number of solutions to equation (\ref{eq:alpeq1}) is $p^{r(n-1)}$. Substituting one such solution into equation (\ref{eq:alpeq2}), we get
\begin{align}
\sum_{j=1}^{n} \alpha_j u_{2j} &= \sum_{j \neq k} \alpha_j u_{2j} - u_{1k}^{-1} u_{2k} \left( \sum_{j \neq k} \alpha_j u_{1j} \right) \\
&= u_{1k}^{-1} \left( \sum_{j \neq k} \alpha_j (u_{1k} u_{2j} - u_{2k} u_{1j}) \right) \\
&= u_{1k}^{-1} \sum_{j \neq k} \alpha_j m_{k,j}
\end{align}

Of the $p^{r(n-1)}$ choices for $\{\alpha_i \}_{i=1}^n$, we need to find those that satisfy $\sum_{j \neq k} \alpha_j m_{k,j} = 0$. We allow $\alpha_j$ to be arbitrary for $j \neq k, j^*$ and solve the equation $\alpha_{j^*} m_{k,j*} = - \sum_{j \neq k,j^*} \alpha_j m_{k,j}$. It is clear that the summation in the right hand side yields a sum that belongs to $p^i \ringr{}$. Since $k,j^*$ are chosen such that $m_{k,j^*} \in p^i \ringr{} \backslash p^{i+1} \ringr{}$, it follows from Lemma \ref{lemma:numsolns} in Appendix \ref{sec:numsolnssec} that this equation has $p^i$ solutions for $\alpha_{j^*}$ for each of the $p^{r(n-2)}$ choices of $\alpha_j, j \neq k,j^*$. Once $\alpha_j, j \neq k$ is fixed, $\alpha_k$ is automatically fixed at $\alpha_k = -u_{1k}^{-1} \sum_{j \neq k} \alpha_j u_{1j}$. Thus, the total number of solutions that simultaneously satisfy equations (\ref{eq:alpeq1}) and (\ref{eq:alpeq2}) is $p^i p^{r(n-2)}$.

It follows that the probability of a randomly chosen homomorphism $\phi_1(\cdot)$ mapping both $u_1^n,u_2^n$ to $0$ is given by $p^i / p^{2r}$. Since each of the $k$ homomorphisms $\phi_i, 1 \leq i \leq k$ can be chosen independently, we have
\begin{equation}
P(\phi(u_1^n) = \phi(u_2^n) = 0) = p^{-(2r-i)k}
\end{equation}
when $D(u_1^n,u_2^n) = p^i \ringr{}$ for some $0 \leq i \leq r$. This proves the claim of Lemma \ref{lemma:jointhomoprob}.
\end{proof}

Suppose $u_1^n$ and $u_2^n$ are non-redundant sequences. It follows from Lemmas \ref{lemma:uniflem} and \ref{lemma:jointhomoprob} that the events $1_{\{u_1^n \in \mc{C} \}}$ and $1_{\{u_2^n \in \mc{C} \}}$ are independent when $D(u_1^n, u_2^n) = \ringr{}$. In order to infer the dependency graph of the indicator random variables in equation (\ref{eq:Thetadef}), we need to count the number of sequences $u_2^n$ for a given $u_1^n$ such that $D(u_1^n,u_2^n) = p^i \ringr{}$ for a given $1 \leq i \leq r$. This is the content of the next lemma.

\begin{lemma} \label{lemma:numu2givenu1}
Let $u_1^n$ be a non-redundant sequence. Let $D_i(u_1^n), 0 \leq i \leq r$ be the set of all $u_2^n$ sequences such that $D(u_1^n,u_2^n) = p^i \ringr{}$. The size of the set $D_i(u_1^n)$ is given by
\begin{equation}
|D_i(u_1^n)|= \left\{ \begin{array}{cc} p^r \left( p^{(r-i)(n-1)} - p^{(r-i-1)(n-1)} \right) & 0 \leq i < r \\ p^r-1 & i = r \end{array} \right.
\end{equation}
\end{lemma}

\begin{proof}[\textbf{Proof}:]
We start by estimating the size of $D_r(u_1^n)$, i.e., the set of $u_2^n$ sequences such that $D(u_1^n,u_2^n) = 0$. Since $D(u_1^n,u_2^n) = 0$, $u_2^n$ must be such that there exists $1 \leq k \leq n$ such that $u_{1k}^{-1}$ exists and $m_{k,j} = 0$ for all $j \neq k$. This implies that $u_{1k} u_{2j} = u_{2k} u_{1j}$ for all $j \neq k$. Define $\eta = u_{1k}^{-1} u_{2k}$. It then follows that $u_{2j} = \eta u_{1j}$ for all $1 \leq j \leq n$ which implies that $u_2^n = \eta u_1^n$ for some $\eta \in \ringr{}$. Since it is assumed that $u_2^n \neq u_1^n$, there are $p^r - 1$ distinct values of $\eta$. Since the sequence $u_1^n$ is non-redundant, it follows that each value of $\eta$ results in a distinct value of $u_2^n$. Thus, $|D_r(u_1^n)| = p^r-1$ as claimed in the Lemma.

Consider the case when $D(u_1^n,u_2^n) = p^i \ringr{}$ for some $0 \leq i < r$. We count the number of $u_2^n$ for a given $u_1^n$ such that $p^i \ringr{}$ is the smallest subgroup containing all the set $M(u_1^n,u_2^n)$. Since $D(u_1^n,u_2^n) = p^i \ringr{}$, $u_2^n$ must be such that there exists $1 \leq k \leq n$ such that $u_{1k}^{-1}$ exists and $m_{k,j^*} \in p^i \ringr{} \backslash p^{i+1} \ringr{}$ for some $1 \leq j^* \leq n, j^* \neq k$. Consider the matrices $M_{k,l}(u_1^n,u_2^n), 1 \leq l \leq n, l \neq k$. Let $\Delta_{k,l} \in p^i \ringr{}, 1 \leq l \leq n, l \neq k$. Fixing the values of the determinants $m_{k,l}(u_1^n,u_2^n)$ to be $\Delta_{k,l}$, we can solve for the entire sequence $u_2^n$. Thus, $D_i(u_1^n)$ contains the union over all permissible values of $\{\Delta_{k,l}\}_{l \neq k}$ of those sequences $u_2^n$ such that $m_{k,l}(u_1^n,u_2^n) = \Delta_{k,l}$ for all $1 \leq l \leq n, l \neq k$.

For a given $\{\Delta_{k,l}\}_{l \neq k}$, let us investigate the number of $u_2^n$ sequences such that $m_{k,l}(u_1^n,u_2^n) = \Delta_{k,l}$ for all $1 \leq l \leq n, l \neq k$. Consider first the equation $m_{k,l^*} = \Delta_{k,l^*}$ for some $l^* \neq k$. Since $u_{1k}$ is invertible, there are $p^r$ possible solutions in $(u_{2k}, u_{2l^*})$ for this equation. Now consider the equations $m_{k,l} , 1 \leq l \leq n, l \neq k,l^*$. Since $u_{2k}$ is already fixed and $u_{1k}$ is invertible, there is precisely one solution to $u_{2l}$ in these equations. Solving these $(n-1)$ equations fixes the sequence $u_2^n$. Thus, the number of solutions to $u_2^n$ for a given $u_1^n$ and $\{\Delta_{k,l}\}_{l \neq k}$ is $p^r$. The number of $\Delta_{k,l}$ such that $\{\Delta_{k,l}\}_{l \neq k} \in p^i \ringr{n-1}$ is clearly $p^{(r-i)(n-1)}$. For $D(u_1^n,u_2^n) = p^i \ringr{}$, there must exist at least one $\Delta_{k,l} \in p^i \ringr{} \backslash p^{i+1} \ringr{}$. The total number of such $\{\Delta_{k,l}\}_{l \neq k}$ is clearly $p^{(r-i)(n-1)} - p^{(r-i-1)(n-1)}$. Putting these arguments together, we get that the size of $D_i(u_1^n)$ is $p^r (p^{(r-i)(n-1)} - p^{(r-i-1)(n-1)})$. This proves the claim of Lemma \ref{lemma:numu2givenu1}.
\end{proof}

We are now ready to infer the dependency graph of the indicator random variables in equation (\ref{eq:Thetadef}). The number of nodes in the dependency graph is $|\typset{}{x^n}|$. Let $I_i$ be the indicator of the event $\{u_i^n \in \mc{C}\}$ and let $I_i$ correspond to the $i$th vertex of the graph. From Lemma \ref{lemma:jointhomoprob}, it follows that vertices $i$ and $j$ are connected (denoted by $i \sim j$) if $D(u_i^n,u_j^n) \neq \ringr{}$. Using Lemma \ref{lemma:numu2givenu1}, the degree of the $i$th vertex can be bounded by $p^{rn} - |D_0(u_1^n)| - 1 = p^{r+(r-1)(n-1)}-1$. Note that this is an upper bound since not all $u_2^n$ sequences counted in Lemma \ref{lemma:numu2givenu1} need belong to $\typset{}{x^n}$.

One version of Suen's inequality can be stated as follows. Let $I_i \in \mbox{Be}(p_i), i \in \mc{I}$ be a family of Bernoulli random variables having a dependency graph L with vertex set $\mc{I}$ and edge set $E(L)$. Let $X = \sum_i I_i$ and $\lambda = \Bbb{E}(X) = \sum_i p_i$. Write $i \sim j$ if $(i,j) \in E(L)$ and let $\Delta = \frac{1}{2} \sum_i \sum_{j \sim i} \Bbb{E}(I_i I_j)$ and $\delta = \max_i \sum_{k \sim i} p_k$. Then
\begin{equation} \label{eq:suenineq}
P(X = 0) \leq \exp \left\{ - \min \left( \frac{\lambda^2}{8 \Delta}, \frac{\lambda}{2}, \frac{\lambda}{6 \delta} \right) \right\}
\end{equation}

Let us estimate the quantities $\lambda, \Delta$ and $\delta$ for our problem. It follows from equation (\ref{eq:Iiprob}) that $\lambda = \Bbb{E}(\Theta(x^n)) = |\typset{}{x^n}| p^{-rk}$. Uniform upper and lower bounds \cite{csiszarbook} exist for the size of the set $\typset{}{x^n}$. An upper bound to $\Delta$ can be established via Lemmas \ref{lemma:jointhomoprob} and \ref{lemma:numu2givenu1} as below.
\begin{align}
\Delta &= \frac{1}{2} \sum_i \sum_{j \sim i} \Bbb{E}(I_i I_j) \\
&= \frac{1}{2} \sum_i \sum_{j \sim i} P(\phi(u_i^n) = \phi(u_j^n) = 0) \\
&= \frac{1}{2} \sum_{u_i^n \in \typset{}{x^n}} \, \sum_{m=1}^r \, \sum_{u_j \in \typset{}{x^n} \cap D_m(u_1^n)} P(\phi(u_i^n) = \phi(u_j^n) = 0) \\
&\stackrel{(a)}{=} \frac{1}{2} \sum_{u_i^n \in \typset{}{x^n}} \, \sum_{m=1}^{r} |\typset{}{x^n} \cap D_m(u_1^n)| \left( \frac{p^m}{p^{2r}} \right)^k \\
&\stackrel{(b)}{\leq} \frac{1}{2} \sum_{u_i^n \in \typset{}{x^n}} \left( (p^r-1)\left(\frac{1}{p^r} \right)^k + \sum_{m=1}^{r-1} p^r \left( p^{(r-m)(n-1)} - p^{(r-m-1)(n-1)} \right) \left( \frac{p^m}{p^{2r}} \right)^k  \right)\\
& = \frac{1}{2} |\typset{}{x^n}| \left( (p^r-1)\left(\frac{1}{p^r} \right)^k + \sum_{m=1}^{r-1} p^r \left( p^{(r-m)(n-1)} - p^{(r-m-1)(n-1)} \right) \left( \frac{p^m}{p^{2r}} \right)^k  \right) \label{eq:Deltabound1}
\end{align}
where $(a)$ follows from Lemma \ref{lemma:jointhomoprob} and $(b)$ follows from Lemma \ref{lemma:numu2givenu1}. This expression can be further simplified by noting that $f(m) \triangleq |D_m(u_1^n)| p^{-k(2r-m)}$ is a decreasing function of $m$. Thus, the summation in the parentheses of equation (\ref{eq:Deltabound1}) can be upper bounded by $(r-1)f(1) = (r-1)|D_1(u_1^n)| p^{-k(2r-1)}$. Thus,
\begin{align}
\Delta &\leq \frac{1}{2} |\typset{}{x^n}| \left( p^{r-rk} + (r-1)p^{nr + k + 1 - n - 2rk}\left(1-\frac{1}{p^{n-1}}\right) \right) \\
&\leq \frac{1}{2} |\typset{}{x^n}| p^{r-rk} \left( 1 + (r-1)p^{(r-1)(n-k-1)} \right) \label{eq:Deltabound}
\end{align}

We now bound the quantity $\delta$.
\begin{align}
\delta &= \max_i \sum_{j \sim i} \Bbb{E}(I_j) \\
&= \max_{u_i^n \in \typset{}{x^n}} \sum_{m=1}^r \sum_{u_j^n \in D_m(u_i^n)} P(\phi(u_j^n) = 0) \\
&\stackrel{(a)}{\leq} \max_{u_i^n \in \typset{}{x^n}} \left(p^{r+(r-1)(n-1)}-1 \right) p^{-rk} \\
&\leq  p^{r(n-k) - (n-1)} \label{eq:deltabound}
\end{align}
where $(a)$ follows from equation (\ref{eq:Iiprob}) and the fact the $P_{U|X}$ is a non-redundant distribution. Using these bounds, we can bound the terms involved in equation (\ref{eq:suenineq}).
\begin{align}
\frac{\lambda^2}{8 \Delta} &\geq \frac{|\typset{}{x^n}|^2 p^{-2rk}}{4 |\typset{}{x^n}| p^{r-rk} \left(1 + (r-1)p^{(r-1)(n-k-1)} \right)} \\
&\geq \frac{|\typset{}{x^n}| p^{-r(k+1)}}{4(1+r p^{(r-1)(n-k-1)})} \\
&\geq \frac{|\typset{}{x^n}|}{8r} p^{-(n(r-1)+k+1)} \label{eq:suenterm2}
\end{align}
where the last inequality holds for sufficiently large $n$. The third term in the exponent in equation (\ref{eq:suenineq}) can be bounded as
\begin{equation} \label{eq:suenterm3}
\frac{\lambda}{6 \delta} \geq \frac{|\typset{}{x^n}|}{6} p^{-(n(r-1)+1)}
\end{equation}

Combining equations (\ref{eq:suenterm2}) and (\ref{eq:suenterm3}), we get a bound on the probability of the event $\{\Theta(x^n) = 0\}$ as
\begin{align} \label{eq:suentheta1}
P(\Theta(x^n) = 0) &\leq \exp \left\{ -\min \left( \frac{|\typset{}{x^n}|}{2}p^{-rk}, \frac{|\typset{}{x^n}|}{8r} p^{-(n(r-1)+k+1)}, \frac{|\typset{}{x^n}|}{6} p^{-(n(r-1)+1)} \right) \right\}
\end{align}
As long as each of the terms in the minimizations goes to $\infty$ as $n \rightarrow \infty$, the probability of not finding a jointly typical sequence with $x^n$ in the codebook $\mc{C}$ goes to $0$. Let $x^n \in \typset{}{X}$ be a typical sequence. It is well known \cite{csiszarbook} that for sufficiently large $n$, the size of the set $\typset{}{x^n}$ is lower bounded as
\begin{equation}
|\typset{}{x^n}| \geq 2^{n(H(U|X) - \epsilon_1(\epsilon))}
\end{equation}
where $\epsilon_1(\epsilon) \rightarrow 0$ as $\epsilon \rightarrow 0$. Therefore,
\begin{align}
\frac{|\typset{}{x^n}|}{2}p^{-rk} &\geq \frac{1}{2} \exp_2 \left(n \left[ H(U|X) - \frac{rk}{n} \log p - \epsilon_1 \right] \right) \label{eq:expbound1}\\
\frac{|\typset{}{x^n}|}{8r} p^{-(n(r-1)+k+1)} &\geq \frac{1}{8r} \exp_2 \left( n \left[ H(U|X) - \left( (r-1) + \frac{k+1}{n} \right) \log p - \epsilon_1 \right]  \right) \label{eq:expbound2}\\
\frac{|\typset{}{x^n}|}{6} p^{-(n(r-1)+1)} &\geq \frac{1}{6} \exp_2 \left( n \left[ H(U|X) - \left( (r-1) + \frac{1}{n} \right) \log p - \epsilon_1 \right] \right) \label{eq:expbound3}
\end{align}
For the probability in equation (\ref{eq:suentheta1}) to decay to $0$, we need the exponents of these three terms to be positive. Equation (\ref{eq:expbound1}) gives us the condition
\begin{equation}
\frac{k}{n} \log p^r < H(U|X)
\end{equation}
while equations (\ref{eq:expbound2}) and (\ref{eq:expbound3}) together give us the condition
\begin{equation}
0 < \frac{k}{n} \log p^r < r(H(U|X) - \log p^{r-1})
\end{equation}
Thus, the dimensionality of the parity check matrix satisfies the asymptotic condition
\begin{equation} \label{eq:kcondnslossy}
\lim_{n \rightarrow \infty} \frac{k(n)}{n} \log p^r = \min( H(U|X), r|H(U|X) - \log p^{r-1}|^{+})
\end{equation}
where $|x|^{+} = \max(x,0)$. Combining these results, we see that provided equation (\ref{eq:kcondnslossy}) is satisfied, $P(\Theta(x^n) = 0)$ goes to $0$ double exponentially. We now show that there exists at least one codebook $\mc{C}$ such that the set $A_{\epsilon}(\mc{C})$ has high probability. We do this by calculating the ensemble average of $P(A_{\epsilon}(\mc{C}))$ over all codebooks $\mc{C}$. It follows from equation (\ref{eq:Aepsprob}) that
\begin{align}
\Bbb{E}(P_X(A_{\epsilon}(\mc{C}))) &\geq \sum_{x^n \in \typset{}{X}} P_X(x^n) P(\Theta(x^n \neq 0))\\
&\geq (1-\epsilon_2) \left( 1 - \exp \left\{ - \min \left( \frac{|\typset{}{x^n}|}{2}p^{-rk}, \frac{|\typset{}{x^n}|}{8r} p^{-(n(r-1)+k+1)}, \frac{|\typset{}{x^n}|}{6} p^{-r(n-1)} \right) \right\} \right)
\end{align}
where $\epsilon_2 \rightarrow 0$ as $n \rightarrow \infty$. Thus, as long as equation (\ref{eq:kcondnslossy}) is satisfied, the expected value of $P_X^n(A_{\epsilon}(\mc{C}))$ can be made arbitrarily close to $1$. This implies that there exists at least one homomorphism such that its kernel is a good source code for the triple $(\mc{X},\mc{U},P_{XU})$.

\section{Good Nested Group Codes} \label{sec:goodsrcchproof}
We now show the existence of good nested group codes satisfying Lemma \ref{lemma:goodnestedcodelemma}.  As was remarked in Definition \ref{defi:nestedgroupcodesdefi}, one way to construct a nested group code is to add rows to the parity check matrix of the fine code to get the parity check matrix of the coarse code. Let the random variables $X,Y,U,V,S$ be as given in Lemma \ref{lemma:goodnestedcodelemma}. Let the parity check matrices of the codes $\mc{C}_{11},\mc{C}_{12}$ and $\mc{C}_2$ be $H_{11}, H_{12}$ and $H_2$ respectively. Let their corresponding dimensions be $k_{11} \times n, k_{12} \times n$ and $k_2 \times n$ respectively. In order to ensure nesting, impose the following structural constraints on these matrices.
\begin{equation}
H_{12} = \left[ \begin{array}{c} H_{11} \\ \Delta H_1 \end{array} \right], \quad H_{2} = \left[ \begin{array}{c} H_{12} \\ \Delta H_2 \end{array} \right]
\end{equation}
Let the dimensions $k_{11}, k_{12}$ and $k_2$ satisfy equations (\ref{eq:nestedeq1}) - (\ref{eq:nestedeq3}).

Generate random $H_2, H_{12}$ matrices by constructing the matrices $H_{11}, \Delta H_1, \Delta H_2$ independently by picking entries uniformly and independently from the group $\ringr{}$. From the proofs in Appendices \ref{sec:goodchproof} and \ref{sec:goodsrcproof}, it follows that the codes $\mc{C}_{11}, \mc{C}_{12}$ and $\mc{C}_2$ are with high probability good source and channel codes respectively for the appropriate triples. By union bound, it follows then that there exists a choice of $H_{11}, \Delta H_1$ and $\Delta H_2$ such that the codebook $\mc{C}_2$ is a good channel code and the nested codes $\mc{C}_{11}$ and $\mc{C}_{12}$ are simultaneously good source codes for their respective triples. This proves the existence of good nested group codes as claimed in Lemma \ref{lemma:goodnestedcodelemma}.

\section{Linear Equations in Groups} \label{sec:numsolnssec}

We now present a lemma on the number of solutions over the group $\ringr{}$ for a linear equation in one variable.
\begin{lemma} \label{lemma:numsolns}
Let $a \in p^i \ringr{} \backslash p^{i+1} \ringr{}$ for some $0 \leq i < r$. Then, the linear equation $ax = b$ has a solution in $x$ if and only if $b \in p^i \ringr{}$. In that case, there are $p^i$ distinct solutions for $x$ over the group $\ringr{}$.
\end{lemma}

\begin{proof}[\textbf{Proof}:] It is clear that the equation $ax = b$ cannot have a solution if $ b \notin p^i \ringr{}$. The rest of the proof proceeds in two stages. We first show that if there exists at least one solution to the equation $ax = b$, then there exists $p^i$ distinct solutions. We then show that at least one solution exists for every $b \in p^i \ringr{}$. Together, these imply Lemma \ref{lemma:numsolns}.

Suppose there exists at least one solution $x_1$ to the equation $ax =
b$. Then, for any $t \in p^{r-i} \ringr{}$, $x_1 + t$ is also a
solution and all such solutions are distinct. Conversely, if $x_1,
x_2$ are both solutions, then $x_1 - x_2 \in p^{r-i} \ringr{}$. Thus,
existence of at least one solution implies the existence of exactly
$p^i$ solutions. Now consider the number of distinct values of the set
$\{ax \colon x \in \ringr{}\}$. Since every distinct value repeats
itself exactly $p^i$ times and there are $p^r$ elements in this set,
it follows that the number of distinct values is $p^{r-i}$. This is
exactly the size of the subgroup $p^i \ringr{}$ which implies that $ax
= b$ has exactly $p^i$ solutions for every element $b \in p^i
\ringr{}$.
\end{proof}

\section{$\mc{T}$ is non-empty} \label{sec:Tnonempty}

Recall the definition of $\mc{T}$ from Section \ref{sec:codingthm} as $\mc{T} = \{ A \colon A \mbox{ is abelian}, |\mc{G}| \leq |A| \leq \alpha \beta, \, G(U,V) \subset A \mbox{ with respect to } P_{UV}\}$. Let $|\mc{U}| = \alpha, |\mc{V}| = \beta$. We now show that the function $G(U,V)$ can always be embedded in some abelian group belonging to $\mc{T}$. Consider the function $G_1(U,V) = (U,V)$. Clearly, $G_1(U,V) \subset \ring{\alpha} \oplus \ring{\beta}$ with respect to $P_{UV}$ for any distribution $P_{UV}$. Since there is an obvious surjective mapping between the functions $G_1(U,V)$ and $G(U,V)$, it follows from Definition \ref{defi:embeddingdefi} that $G(U,V) \subset \ring{\alpha} \oplus \ring{\beta}$ with respect to $P_{UV}$. Since $|\ring{\alpha} \oplus \ring{\beta}| = \alpha \beta$, it follows that this group belongs to the set $\mc{T}$ and hence $\mc{T}$ is always non-empty.

\end{document}